\documentclass[aps,prb,twocolumn,superscriptaddress]{revtex4}
\usepackage{amsmath}
\usepackage{amssymb}
\usepackage{graphicx}
\usepackage{mathtools}
\usepackage{relsize}
\usepackage{lmodern}
\usepackage{slantsc}
\usepackage{scalefnt}
\usepackage{subfigure}
\usepackage{dsfont}
\usepackage{ifpdf}
\usepackage{pgffor}
\usepackage{verbatim}
\usepackage{tikz}
\usetikzlibrary{ }
\usepackage{color}
\usetikzlibrary{arrows,matrix,calc,scopes,decorations.markings}
\allowdisplaybreaks[1]

\newtheorem{theorem}{Theorem}

\newenvironment{proof}[1][Proof]{\noindent\textbf{#1.} }{\ \rule{0.5em}{0.5em}}
\newcommand{\dd}{\mathrm{d}}
\newcommand{\be}{\begin{equation}}
\newcommand{\ee}{\end{equation}}
\newcommand{\bse}{\begin{subequations}}
\newcommand{\ese}{\end{subequations}}
\newcommand{\ket}[1]{|{#1}\rangle}
\newcommand{\bra}[1]{\langle{#1}|}

\newcommand{\Z}{\mathbb{Z}}

\newcommand{\Hil}{\mathcal{H}}

\newcommand{\bpm}{\begin{pmatrix}}
\newcommand{\epm}{\end{pmatrix}}
\newcommand{\bmm}{\begin{matrix}}
\newcommand{\emm}{\end{matrix}}

\newcommand{\Blangle}{\Biggl\langle\bmm} 
\newcommand{\BRvert}{\emm\Biggr\vert} 
\newcommand{\BLvert}{\Biggl\vert\bmm} 
\newcommand{\Brangle}{\emm\Biggr\rangle}

\makeatletter
\newcommand*{\Relbarfill@}{\arrowfill@\Relbar\Relbar\Relbar}
\newcommand*{\xeq}[2][]{\ext@arrow 0055\Relbarfill@{#1}{#2}}
\newcommand{\x}{\times}
\makeatother

\tikzset{->-/.style={decoration={
  markings,
  mark=at position .5 with {\arrow{>}}},postaction={decorate}}}
\tikzset{-<-/.style={decoration={
  markings,
  mark=at position .5 with {\arrow{<}}},postaction={decorate}}}

\usetikzlibrary{decorations.markings}
\usetikzlibrary{decorations.pathmorphing,positioning}
\tikzset{snake it/.style={decorate, decoration={snake,amplitude=0.15mm,segment length=1mm}}}
\tikzset{->-/.style={decoration={
			 markings,
			 mark=at position .55 with {\arrow{latex}}},postaction={decorate}}}

\newcommand{\BdryPachnerMove}{
\begin{tikzpicture}[scale=1]
\draw [] (2.4,8.) -- (3.,8.);
\draw [] (3.,8.) -- (3.6,8.);
\draw [dashed] (1.8,8.) -- (2.4,8.);
\draw [dashed] (3.6,8.) -- (4.2,8.);
\draw [] (2.4,8.) -- (3.,8.6);
\draw [] (3.,8.) -- (3.,8.6);
\draw [] (3.,8.6) -- (3.6,8.);
\draw [very thick, ->] (4.8,8.4) -- (5.6,8.4);
\draw [very thick, ->] (5.6,8.2) -- (4.8,8.2);
\draw [dashed] (6.4,8.) -- (7.,8.);
\draw [dashed] (8.2,8.) -- (8.8,8.);
\draw [] (7.,8.) -- (7.6,8.6);
\draw [] (7.6,8.6) -- (8.2,8.);
\draw [] (7.,8.) -- (8.2,8.);
\end{tikzpicture}
}

\newcommand{\ExpGSbase}{
\Biggl|\;\begin{matrix}
\begin{tikzpicture}[scale=1]
\draw [dashed] (1.6,9.6) -- (2.,9.6);
\draw [] (2.,9.6) -- (2.7,9.6);
\draw [] (2.7,9.6) -- (3.5,9.6);
\draw [dashed] (3.5,9.6) -- (4.,9.6);
\draw [->-] (2.,9.6) -- (3.,10.4);
\draw [->-] (2.7,9.6) -- (3.,10.4);
\draw [->-] (3.5,9.6) -- (3.,10.4);
\node [] at (2.3,10.2) {\scriptsize $a_{1}$};
\node [] at (3.1,9.9) {\scriptsize $a_{2}$};
\node [] at (3.5,10.2) {\scriptsize $a_{3}$};
\end{tikzpicture}\:
\end{matrix}\Biggr\rangle
}
\newcommand{\ExpGSbaseAA}{
\Biggl|\;\begin{matrix}
\begin{tikzpicture}[scale=1]
\draw [dashed] (1.6,9.6) -- (2.,9.6);
\draw [] (2.,9.6) -- (2.7,9.6);
\draw [] (2.7,9.6) -- (3.5,9.6);
\draw [dashed] (3.5,9.6) -- (4.,9.6);
\draw [->-] (2.,9.6) -- (3.,10.4);
\draw [->-] (2.7,9.6) -- (3.,10.4);
\draw [->-] (3.5,9.6) -- (3.,10.4);
\node [] at (3.1,9.9) {\scriptsize $ka_{2}$};
\node [] at (2.3,10.2) {\scriptsize $a_{1}$};
\node [] at (3.5,10.2) {\scriptsize $a_{3}$};
\end{tikzpicture}\:
\end{matrix}\Biggr\rangle
}
\newcommand{\TriangulationMuttation}{
\begin{tikzpicture}[scale=1]
\draw [] (1.8,8.) -- (1.6,8.8);
\draw [] (1.6,8.8) -- (2.2,8.6);
\draw [] (2.2,8.6) -- (3.2,8.6);
\draw [] (3.2,8.6) -- (3.2,9.4);
\draw [] (3.2,9.4) -- (2.4,9.2);
\draw [] (2.4,9.2) -- (1.6,9.6);
\draw [] (1.6,9.6) -- (2.6,10.2);
\draw [] (2.6,10.2) -- (3.6,10.);
\draw [] (3.6,10.) -- (4.,9.2);
\draw [] (4.,9.2) -- (4.,8.4);
\draw [] (4.,8.4) -- (3.4,7.8);
\draw [] (3.4,7.8) -- (2.6,7.8);
\draw [] (1.8,8.) -- (2.2,8.6);
\draw [] (2.2,8.6) -- (2.4,9.2);
\draw [] (2.4,9.2) -- (2.6,10.2);
\draw [] (1.6,9.6) -- (2.2,8.6);
\draw [] (2.2,8.6) -- (2.6,7.8);
\draw [] (2.6,7.8) -- (3.2,8.6);
\draw [] (3.2,8.6) -- (3.4,7.8);
\draw [] (1.8,8.) -- (2.6,7.8);
\draw [] (1.6,8.8) -- (1.6,9.6);
\draw [] (2.4,9.2) -- (3.2,8.6);
\draw [] (2.6,10.2) -- (3.2,9.4);
\draw [] (3.2,9.4) -- (3.6,10.);
\draw [] (3.2,9.4) -- (4.,9.2);
\draw [] (4.,9.2) -- (3.2,8.6);
\draw [] (4.,9.2) -- (3.4,7.8);
\draw [] (6.8,8.) -- (6.6,8.8);
\draw [] (6.6,9.6) -- (7.6,10.2);
\draw [] (7.6,10.2) -- (8.6,10.);
\draw [] (8.6,10.) -- (9.,9.2);
\draw [] (9.,9.2) -- (9.,8.4);
\draw [] (9.,8.4) -- (8.4,7.8);
\draw [] (8.4,7.8) -- (7.6,7.8);
\draw [] (6.8,8.) -- (7.6,7.8);
\draw [] (6.6,8.8) -- (6.6,9.6);
\draw [] (7.6,10.2) -- (7.8,9.);
\draw [] (6.6,9.6) -- (7.8,9.);
\draw [] (6.6,8.8) -- (7.8,9.);
\draw [] (7.8,9.) -- (6.8,8.);
\draw [] (7.8,9.) -- (7.6,7.8);
\draw [] (7.8,9.) -- (8.4,7.8);
\draw [] (9.,8.4) -- (7.8,9.);
\draw [] (7.8,9.) -- (9.,9.2);
\draw [] (8.6,10.) -- (7.8,9.);
\draw [->, very thick] (4.8,9.) -- (5.6,9.);
\node [] at (7.6,9.2) {\scriptsize $0$};
\node [] at (6.4,9.6) {\scriptsize ${1}$};
\node [] at (6.2,8.8) {\scriptsize ${2}$};
\node [] at (6.6,7.8) {\scriptsize ${3}$};
\node [] at (7.4,7.6) {\scriptsize ${4}$};
\node [] at (8.6,7.6) {\scriptsize ${5}$};
\node [] at (9.4,8.4) {\scriptsize ${6}$};
\node [] at (9.2,9.2) {\scriptsize ${7}$};
\node [] at (8.8,10.2) {\scriptsize ${8}$};
\node [] at (7.4,10.4) {\scriptsize ${9}$};
\end{tikzpicture}
}

\newcommand{\tetrahedronSplit}[9]{
\begin{tikzpicture}[scale=1,baseline]
\coordinate (a) at(0,0);
\coordinate (b) at (0.8,-0.75);
\coordinate (c) at (2.5,0);
\coordinate (d) at (1.25,2);
\coordinate (e) at ($ (a) ! 0.3 ! (c) $);
\coordinate (f) at ($ (a) ! 0.45 ! (c) $);
\coordinate (g) at ($ (a) ! 0.5 ! (b) $);
\coordinate (h) at ($ (c) ! 0.5 ! (d) $);
\coordinate (n) at ($ (g) ! 0.5 ! (h) $); 
\draw (a) -- (b) -- (c) -- (d) -- cycle;
\draw (b) -- (d);
\draw (a) -- (e);
\draw (f) -- (c);
\draw[-<,>=latex,line width=0.01pt] (a) -- ($ (a) ! 0.5 ! (b) $);
\draw[-<,>=latex,line width=0.01pt] (b) -- ($ (b) ! 0.5 ! (c) $);
\draw[-<,>=latex,line width=0.01pt] (c) -- ($ (c) ! 0.5 ! (d) $);
\draw[<-,>=latex,line width=0.01pt] (f) -- (c);
\draw[-<,>=latex,line width=0.01pt] (b) -- ($ (b) ! 0.5 ! (d) $);
\draw[-<,>=latex,line width=0.01pt] (a) -- ($ (a) ! 0.5 ! (d) $);

\node[left] at(a) {\scalebox{0.7}{$#1$}};
\node[below] at(b) {\scalebox{0.7}{$#2$}};
\node[right] at(c) {\scalebox{0.7}{$#3$}};
\node[above] at(d) {\scalebox{0.7}{$#4$}};
\node[circle,fill=black,outer sep=0pt, inner sep=0.6pt] at(n) {};
\node[below] at ($ (a) ! 0.5 ! (b) $) {\scalebox{0.7}{$#5$}};
\node[below] at ($ (b) ! 0.5 ! (c) $) {\scalebox{0.7}{$#6$}};
\node[right] at ($ (c) ! 0.5 ! (d) $) {\scalebox{0.7}{$#7$}};
\node at (0.6,-0.15) {\scalebox{0.7}{$#5#6$}};
\node[left] at ($ (b) ! 0.5 ! (d) $) {\scalebox{0.7}{$#6#7$}};
\node[left] at ($ (a) ! 0.5 ! (d) $) {\scalebox{0.7}{$#5#6#7$}};

\draw[dashed] (a) -- (n);
\draw[dashed] (b) -- (n);
\draw[dashed] (c) -- (n);
\draw[dashed] (d) -- (n);

\ifnum #8=1 {
  \node[right, outer sep=0pt, inner sep=1pt] at($ (n) ! 0.07 ! (d) $) {\scalebox{0.7}{$#1'$}};
  \draw[->,>=latex,line width=0.01pt] ($ (a) ! 0.4 ! (n) $) -- ($ (a) ! 0.5 ! (n) $) node[above] {\scalebox{0.7}{$#9$}};
  \draw[->,>=latex,line width=0.01pt] ($ (b) ! 0.4 ! (n) $) node[right] {\scalebox{0.7}{$#9#5$}} -- ($ (b) ! 0.5 ! (n) $) ;
  \draw[->,>=latex,line width=0.01pt] ($ (c) ! 0.4 ! (n) $) -- ($ (c) ! 0.5 ! (n) $) node[above] {\scalebox{0.7}{$#9#5#6$}};
  \draw[->,>=latex,line width=0.01pt] ($ (d) ! 0.4 ! (n) $) -- ($ (d) ! 0.5 ! (n) $) node[right,outer sep=0pt, inner sep=1pt] {\scalebox{0.6}{$#9#5#6#7$}};}
\fi
\ifnum #8=2 {
  \node[right, outer sep=0pt, inner sep=1pt] at($ (n) ! 0.07 ! (d) $) {\scalebox{0.7}{$#2'$}};
  \draw[-<,>=latex,line width=0.01pt] ($ (a) ! 0.4 ! (n) $) -- ($ (a) ! 0.5 ! (n) $) node[above, outer sep=0pt, inner sep=1pt] {\scalebox{0.7}{$#5#9^{-1}$}};
  \draw[->,>=latex,line width=0.01pt] ($ (b) ! 0.4 ! (n) $) node[right] {\scalebox{0.7}{$#9$}} -- ($ (b) ! 0.5 ! (n) $) ;
  \draw[->,>=latex,line width=0.01pt] ($ (c) ! 0.4 ! (n) $) -- ($ (c) ! 0.5 ! (n) $) node[above] {\scalebox{0.7}{$#9#6$}};
  \draw[->,>=latex,line width=0.01pt] ($ (d) ! 0.4 ! (n) $) -- ($ (d) ! 0.5 ! (n) $) node[right,outer sep=0pt, inner sep=1pt] {\scalebox{0.6}{$#9#6#7$}};}
\fi
\ifnum #8=3 {
  \node[right, outer sep=0pt, inner sep=1pt] at($ (n) ! 0.07 ! (d) $) {\scalebox{0.7}{$#3'$}};
  \draw[-<,>=latex,line width=0.01pt] ($ (a) ! 0.4 ! (n) $) -- ($ (a) ! 0.5 ! (n) $) node[above, outer sep=0pt, inner sep=1pt] {\scalebox{0.7}{$#5#6#9^{-1}$}};
  \draw[-<,>=latex,line width=0.01pt] ($ (b) ! 0.4 ! (n) $) node[right, outer sep=0pt, inner sep=1pt] {\scalebox{0.7}{$#6#9^{-1}$}} -- ($ (b) ! 0.5 ! (n) $) ;
  \draw[->,>=latex,line width=0.01pt] ($ (c) ! 0.4 ! (n) $) -- ($ (c) ! 0.5 ! (n) $) node[above] {\scalebox{0.7}{$#9$}};
  \draw[->,>=latex,line width=0.01pt] ($ (d) ! 0.4 ! (n) $) -- ($ (d) ! 0.5 ! (n) $) node[right,outer sep=0pt, inner sep=1pt] {\scalebox{0.6}{$#9#7$}};}
\fi
\ifnum #8=4 {
  \node[right, outer sep=0pt, inner sep=1pt] at($ (n) ! 0.07 ! (d) $) {\scalebox{0.7}{$#4'$}};
  \draw[-<,>=latex,line width=0.01pt] ($ (a) ! 0.4 ! (n) $) -- ($ (a) ! 0.5 ! (n) $) node[above, outer sep=0pt, inner sep=1pt] {\scalebox{0.7}{$#5#6#7#9^{-1}$}};
  \draw[-<,>=latex,line width=0.01pt] ($ (b) ! 0.4 ! (n) $) node[right, outer sep=0pt, inner sep=1pt] {\scalebox{0.7}{$#6#7#9^{-1}$}} -- ($ (b) ! 0.5 ! (n) $) ;
  \draw[-<,>=latex,line width=0.01pt] ($ (c) ! 0.4 ! (n) $) -- ($ (c) ! 0.5 ! (n) $) node[above] {\scalebox{0.7}{$#7#9^{-1}$}};
  \draw[->,>=latex,line width=0.01pt] ($ (d) ! 0.4 ! (n) $) -- ($ (d) ! 0.5 ! (n) $) node[right,outer sep=0pt, inner sep=1pt] {\scalebox{0.6}{$#9$}};}
\fi
\end{tikzpicture}
}

\newcommand{\oneTriangle}[4][0]{
\begin{tikzpicture}[scale=1,baseline]
\draw (0,0) node[left] {\scalebox{0.7}{$#2$}} -- (1,0) node[right] {\scalebox{0.7}{$#3$}} -- (60:1) node[above] {\scalebox{0.7}{$#4$}} -- cycle;
\ifnum #1=1 {
   \draw[-<,>=latex,line width=0.01pt] (0,0) -- (0.5,0) node[below] {\scalebox{0.7}{$[#2#3]$}};
   \draw[-<,>=latex,line width=0.01pt] (0,0) -- (60:0.5) node[left] {\scalebox{0.7}{$[#3#4]$}};
   \draw[->,>=latex,line width=0.01pt] (0,0) +(60:1) -- +(30:0.866) node[right] {\scalebox{0.7}{$[#2#4]$}};}
\else {
   \draw[-<,>=latex,line width=0.01pt] (0,0) -- (0.5,0) ;
   \draw[-<,>=latex,line width=0.01pt] (0,0) -- (60:0.5) ;
   \draw[->,>=latex,line width=0.01pt] (0,0) +(60:1) -- +(30:0.866);}  \fi
\end{tikzpicture}
}

\newcommand{\twoTriangles}[5]{
\begin{tikzpicture}[scale=0.8,baseline]
\coordinate (b) at (0,0);
\coordinate (c) at (30:1);
\coordinate (d) at (0,1);
\coordinate (a) at (150:1);

\draw (a) node[left] {\scalebox{0.7}{$#1$}} -- (b) node[below] {\scalebox{0.7}{$#2$}} -- (c) node[right] {\scalebox{0.7}{$#3$}} -- (d) node[above] {\scalebox{0.7}{$#4$}} -- cycle;
\draw[-<,>=latex,line width=0.01pt] (a) -- ($ (a) ! 0.5 ! (b) $);
\draw[-<,>=latex,line width=0.01pt] (b) -- ($ (b) ! 0.5 ! (c) $);
\draw[-<,>=latex,line width=0.01pt] (c) -- ($ (c) ! 0.5 ! (d) $);
\draw[-<,>=latex,line width=0.01pt] (a) -- ($ (a) ! 0.5 ! (d) $);

\ifnum #5=0 {
   \draw (b)  -- (d);
   \draw[-<,>=latex,line width=0.01pt] (b) -- ($ (b) ! 0.5 ! (d) $);}
\else {
   \draw (a)  -- (c);
   \draw[-<,>=latex,line width=0.01pt] (a) -- ($ (a) ! 0.5 ! (c) $);}
\fi
\end{tikzpicture}
}

\newcommand{\threeTriangles}[5]{
\begin{tikzpicture}[scale=0.8,baseline]
\coordinate (c) at (0,0);
\coordinate (d) at (0,1);
\coordinate (a) at (210:1);
\coordinate (b) at (-30:1);

\draw (a) node[left] {\scalebox{0.7}{$#1$}} -- (b) node[right] {\scalebox{0.7}{$#2$}} -- (c) node[below] {\scalebox{0.7}{$#3$}} -- (d) node[above] {\scalebox{0.7}{$#4$}} -- (a) -- (c);
\draw (b) -- (d);
\draw[-<,>=latex,line width=0.01pt] (a) -- ($ (a) ! 0.5 ! (b) $);
\draw[-<,>=latex,line width=0.01pt] (a) -- ($ (a) ! 0.5 ! (d) $);
\draw[-<,>=latex,line width=0.01pt] (b) -- ($ (b) ! 0.5 ! (d) $);

\ifnum #5=1 {
\draw[-<,>=latex,line width=0.01pt] (b) -- ($ (b) ! 0.5 ! (c) $);
\draw[->,>=latex,line width=0.01pt] (c) -- ($ (c) ! 0.5 ! (d) $);
\draw[-<,>=latex,line width=0.01pt] (a) -- ($ (a) ! 0.5 ! (c) $);}
\fi
\ifnum #5=2 {
\draw[-<,>=latex,line width=0.01pt] (b) -- ($ (b) ! 0.5 ! (c) $);
\draw[-<,>=latex,line width=0.01pt] (c) -- ($ (c) ! 0.5 ! (d) $);
\draw[-<,>=latex,line width=0.01pt] (a) -- ($ (a) ! 0.5 ! (c) $);}
\fi
\ifnum #5=3 {
\draw[->,>=latex,line width=0.01pt] (b) -- ($ (b) ! 0.5 ! (c) $);
\draw[-<,>=latex,line width=0.01pt] (c) -- ($ (c) ! 0.5 ! (d) $);
\draw[-<,>=latex,line width=0.01pt] (a) -- ($ (a) ! 0.5 ! (c) $);}
\fi
\ifnum #5=4 {
\draw[->,>=latex,line width=0.01pt] (b) -- ($ (b) ! 0.5 ! (c) $);
\draw[-<,>=latex,line width=0.01pt] (c) -- ($ (c) ! 0.5 ! (d) $);
\draw[->,>=latex,line width=0.01pt] (a) -- ($ (a) ! 0.5 ! (c) $);}
\fi
\end{tikzpicture}
}

\newcommand{\BvTri}[5]{
\begin{tikzpicture}[scale=1,baseline]
\coordinate (a) at(0,0);    
\coordinate (b) at (0.8,-0.75);  
\coordinate (c) at (2.5,0);      
\coordinate (d) at (1.25,2);     
\coordinate (e) at ($ (a) ! 0.3 ! (c) $);
\coordinate (f) at ($ (a) ! 0.45 ! (c) $);
\coordinate (g) at ($ (a) ! 0.5 ! (b) $);
\coordinate (h) at ($ (c) ! 0.5 ! (d) $);
\coordinate (n) at ($ (g) ! 0.5 ! (h) $); 
\draw (a) -- (b) -- (c) -- (d) -- cycle;
\draw (b) -- (d);
\draw (a) -- (e);
\draw (f) -- (c);
\draw[-<,>=latex,line width=0.01pt] (a) -- ($ (a) ! 0.5 ! (b) $);
\draw[-<,>=latex,line width=0.01pt] (b) -- ($ (b) ! 0.5 ! (c) $);
\draw[-<,>=latex,line width=0.01pt] (c) -- ($ (c) ! 0.5 ! (d) $);
\draw[<-,>=latex,line width=0.01pt] (f) -- (c);
\draw[-<,>=latex,line width=0.01pt] (b) -- ($ (b) ! 0.5 ! (d) $);
\draw[-<,>=latex,line width=0.01pt] (a) -- ($ (a) ! 0.5 ! (d) $);

\node[left] at(a) {\scalebox{0.7}{$#1$}};
\node[below] at(b) {\scalebox{0.7}{$#2$}};
\node[right] at(c) {\scalebox{0.7}{$#4$}};
\node[above] at(d) {\scalebox{0.7}{$#3$}};
\node[circle,fill=black,outer sep=0pt, inner sep=0.6pt] at(n) {};
\draw[dashed] (a) -- (n);
\draw[dashed] (b) -- (n);
\draw[dashed] (c) -- (n);
\draw[dashed] (d) -- (n);

\ifnum #5=1 {
  \node[right, outer sep=0pt, inner sep=1pt] at($ (n) ! 0.07 ! (d) $) {\scalebox{0.7}{$#1'$}};
  \draw[->,>=latex,line width=0.01pt] ($ (a) ! 0.4 ! (n) $) -- ($ (a) ! 0.5 ! (n) $) ;
  \draw[->,>=latex,line width=0.01pt] ($ (b) ! 0.4 ! (n) $) -- ($ (b) ! 0.5 ! (n) $) ;
  \draw[->,>=latex,line width=0.01pt] ($ (c) ! 0.4 ! (n) $) -- ($ (c) ! 0.5 ! (n) $) ;
  \draw[->,>=latex,line width=0.01pt] ($ (d) ! 0.4 ! (n) $) -- ($ (d) ! 0.5 ! (n) $) ;}
\fi
\ifnum #5=2 {
  \node[right, outer sep=0pt, inner sep=1pt] at($ (n) ! 0.07 ! (d) $) {\scalebox{0.7}{$#2'$}};
  \draw[-<,>=latex,line width=0.01pt] ($ (a) ! 0.4 ! (n) $) -- ($ (a) ! 0.5 ! (n) $) ;
  \draw[->,>=latex,line width=0.01pt] ($ (b) ! 0.4 ! (n) $) -- ($ (b) ! 0.5 ! (n) $) ;
  \draw[->,>=latex,line width=0.01pt] ($ (c) ! 0.4 ! (n) $) -- ($ (c) ! 0.5 ! (n) $) ;
  \draw[->,>=latex,line width=0.01pt] ($ (d) ! 0.4 ! (n) $) -- ($ (d) ! 0.5 ! (n) $) ;}
\fi
\ifnum #5=3 {
  \node[right, outer sep=0pt, inner sep=1pt] at($ (n) ! 0.07 ! (d) $) {\scalebox{0.7}{$#3'$}};
  \draw[-<,>=latex,line width=0.01pt] ($ (a) ! 0.4 ! (n) $) -- ($ (a) ! 0.5 ! (n) $) ;
  \draw[-<,>=latex,line width=0.01pt] ($ (b) ! 0.4 ! (n) $) -- ($ (b) ! 0.5 ! (n) $) ;
  \draw[->,>=latex,line width=0.01pt] ($ (c) ! 0.4 ! (n) $) -- ($ (c) ! 0.5 ! (n) $) ;
  \draw[->,>=latex,line width=0.01pt] ($ (d) ! 0.4 ! (n) $) -- ($ (d) ! 0.5 ! (n) $) ;}
\fi
\ifnum #5=4 {
  \node[right, outer sep=0pt, inner sep=1pt] at($ (n) ! 0.07 ! (d) $) {\scalebox{0.7}{$#4'$}};
  \draw[-<,>=latex,line width=0.01pt] ($ (a) ! 0.4 ! (n) $) -- ($ (a) ! 0.5 ! (n) $) ;
  \draw[-<,>=latex,line width=0.01pt] ($ (b) ! 0.4 ! (n) $) -- ($ (b) ! 0.5 ! (n) $) ;
  \draw[->,>=latex,line width=0.01pt] ($ (c) ! 0.4 ! (n) $) -- ($ (c) ! 0.5 ! (n) $) ;
  \draw[-<,>=latex,line width=0.01pt] ($ (d) ! 0.4 ! (n) $) -- ($ (d) ! 0.5 ! (n) $) ;}
\fi
\end{tikzpicture}
}

\newcommand{\cylinderGraph}[5]{
\begin{tikzpicture}[scale=#5,baseline]
\coordinate (p1) at (0,0);
\coordinate (p2) at (0,1);
\coordinate (p3) at (1,0);
\coordinate (p4) at (1,1);
\draw [-,>=latex,line width=0.02pt] (p2) -- (p1);
\draw [-,>=latex,line width=0.02pt] (p4) -- (p3);
\draw [-<-,>=latex,line width=0.02pt] (p3) -- (p1);
\draw [-<-,>=latex,line width=0.02pt] (p4) -- (p2);
\draw [-,>=latex,line width=0.02pt] (p4) -- (p1);
\node at (-0.15,-0.15) {\scalebox{0.7}{1}};
\node at (-0.15,1.15) {\scalebox{0.7}{2}};
\node at (1.15,-0.15) {\scalebox{0.7}{3}};
\node at (1.15,1.15) {\scalebox{0.7}{4}};
\node at (0.5,-0.18) {\scalebox{0.7}{$#1$}};
\node at (0.5,1.18) {\scalebox{0.7}{$#2$}};
\node at (-0.15,0.5) {\scalebox{0.7}{$#3$}};
\node at (1.2,0.5) {\scalebox{0.7}{$#4$}};
\end{tikzpicture}
}

\newcommand{\torusgraphST}[9]{
\begin{tikzpicture}[scale=1,baseline]
\coordinate (p1) at (0,0);
\coordinate (p2) at (0,1);
\coordinate (p3) at (1,0);
\coordinate (p4) at (1,1);
\draw [->-,>=latex,line width=0.02pt] (p2) -- (p1);
\draw [->-,>=latex,line width=0.02pt] (p4) -- (p3);
\draw [->-,>=latex,line width=0.02pt] (p3) -- (p1);
\draw [->-,>=latex,line width=0.02pt] (p4) -- (p2);

\ifnum #9=0 {
    \ifnum #1<#7
       \draw [->-,>=latex,line width=0.02pt] (p4) -- (p1);
    \else   \draw [-<-,>=latex,line width=0.02pt] (p4) -- (p1);
    \fi }
\else {
    \ifnum #3<#5
       \draw [->-,>=latex,line width=0.02pt] (p3) -- (p2);
    \else   \draw [-<-,>=latex,line width=0.02pt] (p3) -- (p2);
    \fi }
\fi
\node at (-0.15,-0.15) {$#1#2$};
\node at (-0.15,1.15) {$#3#4$};
\node at (1.15,-0.15) {$#5#6$};
\node at (1.15,1.15) {$#7#8$};
\end{tikzpicture}
}

\newcommand{\Ygraph}[4][1]{
\begin{tikzpicture}[scale=0.6,baseline]
  \draw [->,>=stealth',line width=0.01pt] (30:0.1) -- (0,0) ; 
    \draw (30:1) -- (0,0) ; 
  \draw [->,>=stealth',line width=0.01pt] (150:0.1) -- (0,0); 
    \draw (150:1) -- (0,0); 
  \node at(-0.5,0.01) {\scalebox{0.8}{$#2$}};
  \node at(0.5,0.01) {\scalebox{0.8}{$#4$}};
  \ifnum #1=1 {
   \draw [->,>=stealth',line width=0.01pt] (0,-1/20) -- (0,0); 
   \node at(-0.2,-0.5) {\scalebox{0.8}{$#3$}};
  }
  \else{
    \draw [<-,>=stealth',line width=0.01pt] (0,-1/2) -- (0,0); 
    \node at(-0.4,-0.5) {\scalebox{0.8}{$-#3$}};
    }
  \fi
    \draw (0,-1) -- (0,0); 
  \end{tikzpicture}
  }

\newcommand{\Hgraph}[3][1]{
  \begin{tikzpicture}[scale=0.6,baseline]
    \coordinate (c) at (0,0);
    \coordinate (l) at (-0.7, 0);
    \coordinate (r) at ($ (c) ! -1 ! (l) $);
    \coordinate (ul) at (-0.95,0.75);
    \coordinate (lr) at ($ (c) ! -1 ! (ul) $);
    \coordinate (ll) at (-0.95,-0.75);
    \coordinate (ur) at ($ (c) ! -1 ! (ll) $);
    \draw (l) -- (r) ; 
    \draw (ul) -- (l); 
    \draw (ll) -- (l); 
    \draw (lr) -- (r); 
    \draw (ur) -- (r); 
    \ifnum #1=1 {
    \node[right] at($ (ul) ! .3 ! (l) $) {\scalebox{0.7}{$#2_1$}};
    \node[right] at($ (ll) ! .2 ! (l) $) {\scalebox{0.7}{$#2_2$}};
    \node[left] at($ (lr) ! .2 ! (r) $) {\scalebox{0.7}{$#2_3$}};
    \node[left] at($ (ur) ! .3 ! (r) $) {\scalebox{0.7}{$#2_4$}};
    \node[below] at($ (c) ! .15 ! (0,0.5) $) {\scalebox{0.7}{$#3$}};}
    \fi
    \draw[<-,>=stealth', line width=0.01pt] (l) -- (r) ;
    \draw[->,>=stealth', line width=0.01pt] (ul) -- (l);
    \draw[->,>=stealth', line width=0.01pt] (ll) -- (l);
    \draw[->,>=stealth', line width=0.01pt] (lr) -- (r);
    \draw[->,>=stealth', line width=0.01pt] (ur) -- (r);
   \end{tikzpicture}
  }

 \newcommand{\Xgraph}[3][1]{
  \begin{tikzpicture}[scale=0.6,baseline]
    \coordinate (c) at (0,0);
    \coordinate (u) at (0, 0.6);
    \coordinate (d) at ($ (c) ! -1 ! (u) $);
    \coordinate (ul) at (-0.85,0.85);
    \coordinate (lr) at ($ (c) ! -1 ! (ul) $);
    \coordinate (ll) at (-0.85,-0.85);
    \coordinate (ur) at ($ (c) ! -1 ! (ll) $);
    \draw (u) -- (d) ; 
    \draw (ul) -- (u); 
    \draw (ll) -- (d); 
    \draw (lr) -- (d); 
    \draw (ur) -- (u); 
    \ifnum #1=1 {
    \node[below] at($ (ul) ! .3 ! (u) $) {\scalebox{0.7}{$#2_1$}};
    \node[above] at($ (ll) ! .2 ! (d) $) {\scalebox{0.7}{$#2_2$}};
    \node[above] at($ (lr) ! .2 ! (d) $) {\scalebox{0.7}{$#2_3$}};
    \node[below] at($ (ur) ! .2 ! (u) $) {\scalebox{0.7}{$#2_4$}};
    \node[right] at($ (-0.5,0) ! 0.8 ! (c) $) {\scalebox{0.7}{$#3$}}; }
    \fi
    \draw[<-,>=stealth', line width=0.01pt] (u) -- (d) ;
    \draw[->,>=stealth', line width=0.01pt] (ul) -- (u);
    \draw[->,>=stealth', line width=0.01pt] (ll) -- (d);
    \draw[->,>=stealth', line width=0.01pt] (lr) -- (d);
    \draw[->,>=stealth', line width=0.01pt] (ur) -- (u);
   \end{tikzpicture}
  }

\newcommand{\Psix}[3][1]{
\begin{tikzpicture}[scale=0.8]
\node[name=s, regular polygon, regular polygon sides=6, minimum size=1cm, outer sep=0pt ,draw] at (0,0) {}; 
\foreach \anchor/\x/\y /\xx/\yy /\b in
{corner 1/0.17/0.17*1.732/-0.11/0.18/1, corner 2/-0.17/0.17*1.732/0.07/0.18/2, corner 3/-0.34/0/-0.15/-0.18/3, corner 4/-0.17/-0.17*1.732/-0.22/-0.05/4, corner 5/0.17/-0.17*1.732/0.2/-0.05/5, corner 6/0.34/0/0.15/-0.18/6}
{
 \draw[shift=(s.\anchor)] (0,0) -- (\x,\y) node at(\xx,\yy) {$#2_{\text{\scalebox{0.7}{$\b$}}}$};
 \ifnum #1=1
 \draw[shift=(s.\anchor),<-,>=stealth', line width=0.01pt] (s.\anchor) -- (\x,\y);
 \fi
 }
\foreach \anchor/\xx/\yy /\a in
{side 1/0/-0.18/1, side 2/-0.18/0.05/2, side 3/0.15/0.05/3, side 4/0/-0.18/4, side 5/-0.18/0.05/5, side 6/0.15/0.05/6}
 \draw[shift=(s.\anchor)]  node at(\xx,\yy) {$#3_{\text{\scalebox{0.7}{$\a$}}}$};
\ifnum #1=1{
  \foreach \anchorr/\anchorf in
   {corner 1/corner 2, corner 2/corner 3, corner 3/corner 4, corner 4/corner 5, corner 5/corner 6, corner 6/corner 1}
   \draw[shift=(s.\anchorr), ->, >=stealth', line width=0.01pt]  (s.\anchorr) -- (s.\anchorf);}
 \else {
  \foreach \anchorb/\anchorw in
   {corner 1/corner 2, corner 3/corner 4, corner 5/corner 6} {
   \node[fill=black, circle, minimum size=2.5, inner sep=0, outer sep=0, draw] at(s.\anchorb) {};
   \node[fill=white, circle, minimum size=2.5, inner sep=0, outer sep=0, draw] at(s.\anchorw) {};}
}
\fi
\end{tikzpicture}
}

\newcommand{\torusHole}{
\begin{tikzpicture}[scale=0.6]
\draw [] (2.8,11.2) -- (3.2,11.);
\draw [] (3.2,11.) -- (3.,11.6);
\draw [] (3.,11.6) -- (3.2,12.);
\draw [] (3.2,12.) -- (3.6,12.);
\draw [] (3.6,12.) -- (3.8,11.6);
\draw [] (3.8,11.6) -- (3.6,11.2);
\draw [] (2.8,11.2) -- (3.,11.6);
\draw [] (3.6,12.) -- (3.8,12.4);
\draw [] (3.8,12.4) -- (4.,12.);
\draw [] (3.6,12.) -- (4.,12.);
\draw [] (3.8,11.6) -- (4.4,11.6);
\draw [] (4.,12.) -- (3.8,11.6);
\draw [] (4.,12.) -- (4.4,11.6);
\draw [] (4.4,11.6) -- (4.,11.2);
\draw [] (4.,11.2) -- (3.8,11.6);
\draw [] (3.6,11.2) -- (4.,11.2);
\draw [] (3.2,11.) -- (3.6,11.2);
\draw [] (2.4,11.6) -- (2.8,11.2);
\draw [] (2.8,11.2) -- (3.2,10.6);
\draw [] (3.2,11.) -- (3.2,10.6);
\draw [] (3.8,12.4) -- (4.2,12.4);
\draw [] (4.,12.) -- (4.4,12.);
\draw [] (4.,11.2) -- (4.2,10.8);
\draw [] (4.,11.2) -- (4.4,11.2);
\draw [] (2.6,10.8) -- (2.8,11.2);
\draw [] (2.8,11.2) -- (2.4,11.2);
\draw [] (3.8,12.4) -- (4.,12.8);
\draw [] (3.2,10.6) -- (2.8,10.4);
\draw [very thick] (3.2,12.) -- (3.,11.6);
\draw [very thick] (3.,11.6) -- (3.2,11.);
\draw [very thick] (3.2,11.) -- (3.6,11.2);
\draw [very thick] (3.6,11.2) -- (3.8,11.6);
\draw [very thick] (3.8,11.6) -- (3.6,12.);
\draw [very thick] (3.2,12.) -- (3.6,12.);
\draw [very thick] (11.,9.4) -- (10.6,9.2);
\draw [very thick] (10.6,9.2) -- (10.4,9.6);
\draw [very thick] (10.4,9.6) -- (10.6,10.);
\draw [very thick] (10.6,10.) -- (11.2,9.8);
\draw [very thick] (11.2,9.8) -- (11.,9.4);
\draw [] (10.6,10.) -- (10.8,10.4);
\draw [] (10.8,10.4) -- (11.2,9.8);
\draw [] (11.2,9.8) -- (11.6,9.8);
\draw [] (11.6,9.8) -- (11.,9.4);
\draw [] (11.,9.4) -- (11.,9.);
\draw [] (11.,9.) -- (10.6,9.2);
\draw [] (10.6,9.2) -- (10.2,9.);
\draw [] (10.2,9.) -- (10.4,9.6);
\draw [] (10.4,9.6) -- (10.,10.);
\draw [] (10.,10.) -- (10.6,10.);
\draw [] (10.2,9.) -- (10.,10.);
\draw [] (10.,10.) -- (10.8,10.4);
\draw [] (10.8,10.4) -- (11.6,9.8);
\draw [] (11.6,9.8) -- (11.,9.);
\draw [] (11.,9.) -- (10.2,9.);
\draw [] (10.8,10.4) -- (11.2,10.6);
\draw [] (10.,10.) -- (10.,10.6);
\draw [] (10.,10.) -- (9.6,9.8);
\draw [] (9.6,9.) -- (10.2,9.);
\draw [] (10.4,8.6) -- (11.,9.);
\draw [] (11.4,9.) -- (11.6,9.8);
\draw [] (11.6,9.8) -- (11.8,10.2);
\draw [] (2.4,11.6) -- (2.6,12.2);
\draw [] (2.4,11.6) -- (3.,11.6);
\draw [] (3.2,12.) -- (2.6,12.2);
\draw [] (3.2,12.) -- (3.2,12.6);
\draw [] (3.6,12.) -- (3.2,12.6);
\draw [] (3.2,12.6) -- (3.8,12.4);
\draw [] (2.4,11.6) -- (3.2,12.);
\draw [] (2.6,12.2) -- (3.2,12.6);
\draw [] (3.2,12.6) -- (3.6,13.);
\draw [] (3.2,10.6) -- (3.6,11.2);
\draw [] (3.6,11.2) -- (3.8,10.8);
\draw [] (3.2,10.6) -- (3.8,10.8);
\draw [] (3.8,10.8) -- (3.6,10.4);
\draw [] (4.,11.2) -- (3.8,10.8);
\draw [] (3.8,10.8) -- (4.,10.6);
\draw [very thick] (6.2,8.8) -- (6.,8.4);
\draw [very thick] (6.,8.4) -- (6.6,8.);
\draw [very thick] (6.6,8.) -- (6.8,8.6);
\draw [very thick] (6.8,8.6) -- (6.2,8.8);
\draw [] (6.,8.4) -- (5.6,8.6);
\draw [] (5.6,8.6) -- (5.8,9.);
\draw [] (5.8,9.) -- (6.2,8.8);
\draw [] (6.2,8.8) -- (6.6,9.2);
\draw [] (6.6,9.2) -- (7.2,8.6);
\draw [] (7.2,8.6) -- (6.8,8.6);
\draw [] (7.2,8.6) -- (6.8,7.8);
\draw [] (6.6,8.) -- (6.8,7.8);
\draw [] (6.8,7.8) -- (6.,7.8);
\draw [] (6.,7.8) -- (6.6,8.);
\draw [] (6.,8.4) -- (6.,7.8);
\draw [] (6.,7.8) -- (5.6,8.);
\draw [] (5.6,8.) -- (6.,8.4);
\draw [] (6.6,9.2) -- (6.8,8.6);
\draw [] (5.6,8.6) -- (5.,8.4);
\draw [] (6.6,9.2) -- (6.2,9.2);
\draw [] (7.2,8.6) -- (7.4,8.8);
\draw [] (6.8,7.8) -- (7.2,7.8);
\draw [] (6.,7.8) -- (6.2,7.6);
\draw [] (6.6,8.) -- (7.2,8.6);
\draw [] (5.6,8.) -- (5.6,8.6);
\draw [] (5.6,8.) -- (5.4,7.6);
\draw [] (6.2,8.8) -- (6.2,9.2);
\draw [] (6.2,9.2) -- (5.8,9.4);
\draw [] (6.6,9.2) -- (6.4,9.6);
\draw [] (2.6,12.2) -- (2.4,12.6);
\draw [] (6.,8.4) -- (5.8,9.);
\draw [] (5.8,9.) -- (5.4,9.2);
\draw [gray] (1.2,12.6) ..controls (0.757,11.844) and (0.953,10.66) .. (1.8,9.6)
                         ..controls (2.647,8.54) and (4.145,7.604) .. (5.6,7.)
                         ..controls (7.055,6.396) and (8.467,6.124) .. (9.8,6.2)
                         ..controls (11.133,6.276) and (12.388,6.699) .. (13.,7.4)
                         ..controls (13.612,8.101) and (13.58,9.078) .. (13.,10.)
                         ..controls (12.42,10.922) and (11.29,11.787) .. (10.2,12.4)
                         ..controls (9.11,13.013) and (8.059,13.374) .. (7.,13.6)
                         ..controls (5.941,13.826) and (4.875,13.918) .. (3.8,13.8)
                         ..controls (2.725,13.682) and (1.643,13.356) .. (1.2,12.6);
\draw [gray] (5.2,11.4) ..controls (5.176,11.224) and (5.152,11.048) .. (5.4,10.8)
                         ..controls (5.648,10.552) and (6.167,10.233) .. (6.8,10.)
                         ..controls (7.433,9.767) and (8.181,9.619) .. (8.6,9.6)
                         ..controls (9.019,9.581) and (9.11,9.69) .. (9.2,9.8);
\draw [gray] (5.4,10.8) ..controls (5.631,10.933) and (5.862,11.067) .. (6.4,11.)
                         ..controls (6.938,10.933) and (7.782,10.667) .. (8.2,10.4)
                         ..controls (8.618,10.133) and (8.609,9.867) .. (8.6,9.6);
\node [] at (3.385,11.6) {\scalebox{0.7}{ $\partial_1\Gamma$}};
\node [] at (6.37,8.42) {\scalebox{0.7}{ $\partial_2\Gamma$}};
\node [] at (10.8,9.6) {\scalebox{0.7} {$\partial_3\Gamma$}};
\end{tikzpicture}
}

\newcommand{\bPlaquette}[6]{
\begin{tikzpicture}[scale=1]
\draw [thick, very thick] (1.6,12.) -- (2.8,12.);
\draw [thick, very thick, dashed] (0.8,12.) -- (1.6,12.);
\draw [thick, very thick, dashed] (2.8,12.) -- (3.6,12.);
\draw [dotted] (1.6,12.) -- (2.2,11.2);
\draw [dotted] (2.2,11.2) -- (2.8,12.);
\draw [thick] (1.6,12.) -- (2.2,12.8);
\draw [thick] (2.2,12.8) -- (2.8,12.);
\node [] at (1.6,11.8) {\scriptsize $#1$};
\node [] at (2.8,11.8) {\scriptsize $#2$};
\node [] at (2.4,12.8) {\scriptsize $#3$};
\node [] at (3.2,12.4) {\scriptsize $#4$};
\node [right] at (3.6,12.) {\scriptsize $#5$};
\node [] at (2.2,11.8) {\scriptsize $#6$};
\end{tikzpicture}
}

\newcommand{\bPlaquetteTwo}[4]{
\begin{tikzpicture}[scale=1]
\draw [thick] (1.6,12.8) -- (2.8,12.8);
\draw [thick, thin] (2.4,13.8) -- (1.6,12.8);
\draw [thick, thin] (2.4,13.8) -- (2.8,12.8);
\draw [thick, thin] (2.4,13.8) -- (2.2,12.8);
\draw [thick, dashed] (3.2,12.8) -- (2.8,12.8);
\draw [thick, dashed] (1.2,12.8) -- (1.6,12.8);
\node [] at (2.6,13.8) {\scriptsize $#1$};
\node [] at (1.6,12.6) {\scriptsize $#2$};
\node [] at (2.2,12.6) {\scriptsize $#3$};
\node [] at (2.8,12.6) {\scriptsize $#4$};
\node [] at (2.9,13.4) {\scriptsize $bulk$};
\end{tikzpicture}
}

\newcommand{\AvAmplitude}[5]{
\begin{tikzpicture}[scale=1]
\draw [thick, dashed] (0.8,12.) -- (1.2,12.);
\draw [thick, dashed] (2.4,12.) -- (2.8,12.);
\draw [thin] (1.2,12.) -- (2.,13.);
\draw [thin] (2.,13.) -- (2.4,12.);
\draw [thick] (2.,11.4) -- (2.4,12.);
\draw [thin] (2.,13.) -- (2.,12.1);
\draw [thin] (2.,11.4) -- (2.,11.9);
\draw [thick] (1.2,12.) -- (2.,11.4);
\draw [thick] (1.2,12.) -- (2.4,12.);
\draw [thin] (2.,13.) -- (1.8,12.);
\draw [thin] (1.8,12.) -- (2.,11.4);
\path [gray, fill, opacity=0.5] (2.,13.) -- (1.8,12.) -- (2.,11.4) -- cycle;
\path [gray, fill, opacity=0.5] (1.2,12.) -- (2.,11.4) -- (1.8,12.) -- cycle;
\path [gray, fill, opacity=0.5] (1.8,12.) -- (2.4,12.) -- (2.,11.4) -- cycle;
\node [] at (2.2,11.4) {\scriptsize $#5$};
\node [] at (1.2,12.2) {\scriptsize $#1$};
\node [] at (1.7,12.2) {\scriptsize $#2$};
\node [] at (2.5,12.2) {\scriptsize $#3$};
\node [] at (2.2,13.) {\scriptsize $#4$};
\end{tikzpicture}
}

\newcommand{\AvCommute}{
\begin{tikzpicture}[scale=1]

\path [gray, fill, opacity=0.5] (5.,13.8) -- (4.8,12.8) -- (5.,12.2) -- cycle;
\path [gray, fill, opacity=0.5] (4.2,12.8) -- (5.,12.2) -- (4.8,12.8) -- cycle;
\path [gray, fill, opacity=0.5] (4.8,12.8) -- (5.4,12.8) -- (5.,12.2) -- cycle;
\path [gray, fill, opacity=0.5] (7.4,11.8) -- (7.2,10.8) -- (7.4,10.2) -- cycle;
\path [gray, fill, opacity=0.5] (6.6,10.8) -- (7.4,10.2) -- (7.2,10.8) -- cycle;
\path [gray, fill, opacity=0.5] (7.2,10.8) -- (7.8,10.8) -- (7.4,10.2) -- cycle;
\draw [thick, dashed] (3.8,12.8) -- (4.2,12.8);
\draw [thick, dashed] (5.4,12.8) -- (5.8,12.8);
\draw [thin] (4.2,12.8) -- (5.,13.8);
\draw [thin] (5.,13.8) -- (5.4,12.8);
\draw [thick] (5.,12.2) -- (5.4,12.8);
\draw [thin] (5.,13.8) -- (5.,12.9);
\draw [thin] (5.,12.2) -- (5.,12.7);
\draw [thick] (4.2,12.8) -- (5.,12.2);
\draw [thick] (4.2,12.8) -- (5.4,12.8);
\draw [thin] (5.,13.8) -- (4.8,12.8);
\draw [thin] (4.8,12.8) -- (5.,12.2);
\draw [thick, dashed] (1.4,12.8) -- (1.8,12.8);
\draw [thick, dashed] (3.,12.8) -- (3.4,12.8);
\draw [thin] (1.8,12.8) -- (2.6,13.8);
\draw [thin] (2.6,13.8) -- (3.,12.8);
\draw [thick] (1.8,12.8) -- (3.,12.8);
\draw [thin] (2.6,13.8) -- (2.4,12.8);
\draw [thick, dashed] (6.2,12.8) -- (6.6,12.8);
\draw [thick, dashed] (7.8,12.8) -- (8.2,12.8);
\draw [thin] (6.6,12.8) -- (7.4,13.8);
\draw [thin] (7.4,13.8) -- (7.8,12.8);
\draw [thin] (7.4,13.8) -- (7.2,12.8);
\draw [thick] (6.6,12.8) -- (6.6,12.2);
\draw [thick] (6.6,12.2) -- (7.2,12.8);
\draw [thin] (6.6,12.2) -- (7.4,13.8);
\draw [thick] (7.2,12.8) -- (7.8,12.8);
\draw [thick, dashed] (6.2,10.8) -- (6.6,10.8);
\draw [thick, dashed] (7.8,10.8) -- (8.2,10.8);
\draw [thin] (6.6,10.8) -- (7.4,11.8);
\draw [thin] (7.4,11.8) -- (7.8,10.8);
\draw [thick] (7.4,10.2) -- (7.8,10.8);
\draw [thin] (7.4,11.8) -- (7.4,10.9);
\draw [thin] (7.4,10.2) -- (7.4,10.7);
\draw [thick] (6.6,10.8) -- (7.4,10.2);
\draw [thick] (6.6,10.8) -- (7.8,10.8);
\draw [thin] (7.4,11.8) -- (7.2,10.8);
\draw [thin] (7.2,10.8) -- (7.4,10.2);
\draw [thick, dashed] (1.4,10.8) -- (1.8,10.8);
\draw [thick, dashed] (3.,10.8) -- (3.4,10.8);
\draw [thin] (1.8,10.8) -- (2.6,11.8);
\draw [thin] (2.6,11.8) -- (3.,10.8);
\draw [thick] (1.8,10.8) -- (3.,10.8);
\draw [thin] (2.6,11.8) -- (2.4,10.8);
\draw [thick, dashed] (5.4,10.8) -- (5.8,10.8);
\draw [thin] (4.2,10.8) -- (5.,11.8);
\draw [thin] (5.,11.8) -- (5.4,10.8);
\draw [thin] (5.,11.8) -- (4.8,10.8);
\draw [thick] (4.2,10.2) -- (4.8,10.8);
\draw [->, thick] (3.4,13.4) -- (4.,13.4);
\draw [->, thick] (5.8,13.4) -- (6.4,13.4);
\draw [->, thick] (3.4,11.4) -- (4.,11.4);
\draw [->, thick] (5.8,11.4) -- (6.4,11.4);
\draw [thick] (4.2,10.2) -- (4.2,10.8);
\draw [dashed, thick] (3.8,10.8) -- (4.2,10.8);
\draw [thin] (4.2,10.2) -- (5.,11.8);
\draw [thick] (4.8,10.8) -- (5.4,10.8);
\node [] at (1.8,12.6) {\scriptsize $2$};
\node [] at (2.4,12.6) {\scriptsize $1$};
\node [] at (3.,12.6) {\scriptsize $3$};
\node [] at (2.8,13.8) {\scriptsize $0$};
\node [] at (4.2,12.6) {\scriptsize $2$};
\node [] at (5.4,12.6) {\scriptsize $3$};
\node [] at (4.7,12.9) {\scriptsize $1$};
\node [] at (5.2,12.2) {\scriptsize $1'$};
\node [] at (5.2,13.8) {\scriptsize $0$};
\node [] at (7.8,12.6) {\scriptsize $3$};
\node [] at (7.6,13.8) {\scriptsize $0$};
\node [] at (6.5,12.9) {\scriptsize $2$};
\node [] at (6.8,12.2) {\scriptsize $2'$};
\node [] at (7.2,12.6) {\scriptsize $1'$};
\node [] at (1.8,10.6) {\scriptsize $2$};
\node [] at (2.4,10.6) {\scriptsize $1$};
\node [] at (3.,10.6) {\scriptsize $3$};
\node [] at (2.8,11.8) {\scriptsize $0$};
\node [] at (7.8,10.6) {\scriptsize $3$};
\node [] at (7.1,10.9) {\scriptsize $1$};
\node [] at (7.6,10.2) {\scriptsize $1'$};
\node [] at (7.6,11.8) {\scriptsize $0$};
\node [] at (5.4,10.6) {\scriptsize $3$};
\node [] at (5.2,11.8) {\scriptsize $0$};
\node [] at (4.1,10.9) {\scriptsize $2$};
\node [] at (4.4,10.2) {\scriptsize $2'$};
\node [] at (4.8,10.6) {\scriptsize $1$};
\node [] at (6.6,10.6) {\scriptsize $2'$};
\node [] at (3.7,11.6) {\scriptsize $A^k_2$};
\node [] at (6.1,11.6) {\scriptsize $A^{k'}_1$};
\node [] at (3.7,13.6) {\scriptsize $A^{k'}_1$};
\node [] at (6.1,13.6) {\scriptsize $A^k_2$};
\path [white, fill] (6.8,12.8) -- (6.8,12.7) -- (6.9,12.7) -- (6.9,12.8) -- cycle;
\path [white, fill] (7.,12.8) -- (7.,12.9) -- (6.9,12.9) -- (6.9,12.8) -- cycle;
\path [fill, opacity=0.3] (6.6,12.2) -- (7.4,13.8) -- (7.2,12.8) -- cycle;
\path [white, fill] (4.4,10.8) -- (4.4,10.7) -- (4.5,10.7) -- (4.5,10.8) -- cycle;
\path [white, fill] (4.6,10.8) -- (4.6,10.9) -- (4.5,10.9) -- (4.5,10.8) -- cycle;
\path [fill, opacity=0.3] (4.2,10.2) -- (5.,11.8) -- (4.8,10.8) -- cycle;
\path [fill, opacity=0.3] (6.6,12.8) -- (6.6,12.2) -- (7.2,12.8) -- cycle;
\path [fill, opacity=0.3] (4.2,10.8) -- (4.2,10.2) -- (4.8,10.8) -- cycle;
\draw [thick] (6.6,12.8) -- (7.2,12.8);
\draw [thick] (4.2,10.8) -- (4.8,10.8);
\end{tikzpicture}
}

\newcommand{\Frobenius}{
\begin{tikzpicture}[scale=1]

\path [gray, fill, opacity=0.3] (3.2,13.6) -- (3.2,12.4) -- (4.4,12.4) -- cycle;
\path [lightgray, fill, opacity=0.5] (3.2,13.6) -- (4.4,12.4) -- (4.4,13.6) -- cycle;
\draw [thick] (3.2,12.4) -- (3.2,13.6);
\draw [thick] (3.2,13.6) -- (4.4,13.6);
\draw [thick] (4.4,13.6) -- (4.4,12.4);
\draw [thick] (4.4,12.4) -- (3.2,12.4);
\draw [thick] (3.2,13.6) -- (4.4,12.4);
\draw [thick] (3.2,12.4) -- (3.7,12.9);
\draw [thick] (3.9,13.1) -- (4.4,13.6);
\node [] at (3.,13.6) {\scriptsize $2$};
\node [] at (4.6,13.6) {\scriptsize $1$};
\node [] at (4.6,12.4) {\scriptsize $1'$};
\node [] at (3.,12.4) {\scriptsize $2'$};

\end{tikzpicture}
}

\newcommand{\cylinderGSD}{
\begin{tikzpicture}[scale=1]

\path [gray, fill, opacity=0.5, opacity=0.3] (2.4,10.4) -- (3.6,11.4) -- (6.,11.4) -- cycle;
\path [gray, fill, opacity=0.5, opacity=0.3] (2.4,10.4) -- (4.8,10.4) -- (6.,11.4) -- cycle;
\draw [thick] (2.4,11.8) -- (3.6,12.8);
\draw [thick] (6.,11.4) -- (4.8,10.4);
\draw [thick] (6.,12.8) -- (4.8,11.8);
\draw [thick] (2.4,10.4) -- (3.307,11.156);
\draw [thick] (3.461,11.284) -- (3.6,11.4);
\draw [thin] (2.4,10.4) -- (4.8,10.4);
\draw [thin] (2.4,11.8) -- (2.4,10.4);
\draw [thin] (2.4,11.8) -- (3.6,11.4);
\draw [thin] (2.4,11.8) -- (4.8,10.4);
\draw [thin] (2.4,11.8) -- (4.8,11.8);
\draw [thin] (2.4,11.8) -- (6.,12.8);
\draw [thin] (3.6,12.8) -- (6.,12.8);
\draw [thin] (4.8,11.8) -- (4.8,10.4);
\draw [thin] (4.8,11.8) -- (6.,11.4);
\draw [thin] (6.,12.8) -- (6.,11.4);
\draw [thin] (2.4,11.8) -- (4.677,11.547);
\draw [thin] (4.875,11.525) -- (6.,11.4);
\draw [thin] (3.6,12.8) -- (4.282,12.402);
\draw [thin] (4.454,12.302) -- (4.906,12.038);
\draw [thin] (5.078,11.938) -- (6.,11.4);
\draw [thin] (2.4,10.4) -- (3.924,10.823);
\draw [thin] (4.116,10.877) -- (4.68,11.033);
\draw [thin] (4.872,11.087) -- (6.,11.4);
\draw [thin] (3.6,12.8) -- (3.6,12.2);
\draw [thin] (3.6,12.) -- (3.6,11.836);
\draw [thin] (3.6,11.636) -- (3.6,11.4);
\draw [thin] (3.6,11.4) -- (4.7,11.4);
\draw [thin] (4.9,11.4) -- (6.,11.4);
\node [] at (2.3,11.9) {\scriptsize $1$};
\node [] at (3.4,12.8) {\scriptsize $2$};
\node [] at (4.7,11.9) {\scriptsize $3$};
\node [] at (6.2,12.8) {\scriptsize $4$};
\node [] at (6.2,11.4) {\scriptsize $4'$};
\node [] at (5.,10.4) {\scriptsize $3'$};
\node [] at (3.7,11.2) {\scriptsize $2'$};
\node [] at (2.3,10.4) {\scriptsize $1'$};
\node [] at (4.7,12.7) {\scriptsize $g$};
\node [] at (4.2,12.) {\scriptsize $g$};
\node [] at (4.,10.6) {\scriptsize $g'$};
\node [] at (2.8,12.4) {\scriptsize $k_1$};
\node [] at (5.5,12.1) {\scriptsize $k_2$};
\node [] at (5.6,10.8) {\scriptsize $k'_2$};
\node [] at (4.5,11.3) {\scriptsize $g$};
\node [] at (3.,11.1) {\scriptsize $k'_1$};
\node [] at (2.3,11.2) {\scriptsize $x$};
\node [] at (3.7,12.5) {\scriptsize $x$};
\node [] at (6.1,12.1) {\scriptsize $y$};
\node [] at (4.9,10.8) {\scriptsize $y$};
\path [gray, fill, opacity=0.5, opacity=0.3] (2.4,11.8) -- (3.6,12.8) -- (6.,12.8) -- cycle;
\path [gray, fill, opacity=0.5, opacity=0.3] (2.4,11.8) -- (4.8,11.8) -- (6.,12.8) -- cycle;
\end{tikzpicture}
}

\newcommand{\cylinderGSDsimple}{
\begin{tikzpicture}[scale=1]
\path [gray, fill, opacity=0.5, opacity=0.3] (1.7,8.) -- (2.9,9.) -- (5.3,9.) -- cycle;
\path [gray, fill, opacity=0.5, opacity=0.3] (1.7,8.) -- (4.1,8.) -- (5.3,9.) -- cycle;
\draw [] (1.7,9.4) -- (4.1,9.4);
\draw [] (4.1,8.) -- (5.3,10.4);
\draw [] (1.7,9.4) -- (4.013,9.143);
\draw [] (4.211,9.121) -- (4.517,9.087);
\draw [] (4.715,9.065) -- (5.3,9.);
\draw [] (2.9,10.4) -- (3.582,10.002);
\draw [] (3.754,9.902) -- (4.23,9.624);
\draw [] (4.402,9.524) -- (4.662,9.372);
\draw [] (4.834,9.272) -- (5.3,9.);
\draw [] (1.7,8.) -- (2.195,8.991);
\draw [] (2.285,9.169) -- (2.339,9.279);
\draw [] (2.429,9.457) -- (2.471,9.543);
\draw [] (2.561,9.721) -- (2.9,10.4);
\draw [thick] (1.7,9.4) -- (2.9,10.4);
\draw [thick] (5.3,9.) -- (4.1,8.);
\draw [thick] (5.3,10.4) -- (4.1,9.4);
\draw [thick] (1.7,8.) -- (2.607,8.756);
\draw [thick] (2.761,8.884) -- (2.9,9.);
\draw [thin] (1.7,8.) -- (4.1,8.);
\draw [thin] (1.7,9.4) -- (1.7,8.);
\draw [thin] (1.7,9.4) -- (4.1,8.);
\draw [thin] (1.7,9.4) -- (5.3,10.4);
\draw [thin] (2.9,10.4) -- (5.3,10.4);
\draw [thin] (4.1,9.4) -- (4.1,8.);
\draw [thin] (5.3,10.4) -- (5.3,9.);
\draw [thin] (2.9,9.) -- (4.,9.);
\draw [thin] (4.2,9.) -- (5.3,9.);
\draw [thin] (2.9,10.4) -- (2.9,9.8);
\draw [thin] (2.9,9.6) -- (2.9,9.436);
\draw [thin] (2.9,9.236) -- (2.9,9.);
\draw [thin] (1.7,8.) -- (3.224,8.423);
\draw [thin] (3.416,8.477) -- (3.98,8.633);
\draw [thin] (4.172,8.687) -- (4.376,8.743);
\draw [thin] (4.568,8.797) -- (5.3,9.);
\node [] at (1.6,9.5) {\scriptsize $1$};
\node [] at (1.6,8.) {\scriptsize $1'$};
\node [] at (2.7,10.4) {\scriptsize $2$};
\node [] at (2.8,9.1) {\scriptsize $2'$};
\node [] at (4.,9.5) {\scriptsize $3$};
\node [] at (4.3,8.) {\scriptsize $3'$};
\node [] at (5.5,10.4) {\scriptsize $4$};
\node [] at (5.5,9.) {\scriptsize $4'$};
\node [] at (3.8,8.9) {\scriptsize $g$};
\node [] at (4.,10.3) {\scriptsize $g$};
\node [] at (2.1,10.) {\scriptsize $k_1$};
\node [] at (4.9,8.4) {\scriptsize $k_2$};
\node [] at (1.6,8.8) {\scriptsize $x$};
\node [] at (2.8,9.9) {\scriptsize $x$};
\node [] at (5.4,9.7) {\scriptsize $y$};
\node [] at (5.1,9.5) {\scriptsize $yk_2$};
\node [] at (3.2,8.7) {\scriptsize $g\bar y$};
\node [] at (2.9,7.8) {\scriptsize $xg\bar y$};
\node [] at (2.3,8.7) {\scriptsize $k_1$};
\path [gray, fill, opacity=0.5, opacity=0.3] (1.7,9.4) -- (2.9,10.4) -- (5.3,10.4) -- cycle;
\path [gray, fill, opacity=0.5, opacity=0.3] (1.7,9.4) -- (4.1,9.4) -- (5.3,10.4) -- cycle;
\end{tikzpicture}
}

\newcommand{\cylinderGSDsplit}{
\begin{tikzpicture}[scale=1]
\path [gray, fill, opacity=0.5, opacity=0.3] (1.7,8.) -- (2.9,9.) -- (5.3,9.) -- cycle;
\path [gray, fill, opacity=0.5, opacity=0.3] (1.7,8.) -- (4.1,8.) -- (5.3,9.) -- cycle;
\draw [] (1.7,9.4) -- (5.3,9.);
\draw [] (2.9,10.4) -- (5.3,9.);
\draw [] (5.2,9.4) -- (7.6,9.4);
\draw [] (7.6,8.) -- (8.8,10.4);
\draw [] (1.7,8.) -- (2.195,8.991);
\draw [] (2.285,9.169) -- (2.327,9.255);
\draw [] (2.417,9.433) -- (2.9,10.4);
\draw [] (5.2,9.4) -- (7.513,9.143);
\draw [] (7.711,9.121) -- (8.017,9.087);
\draw [] (8.215,9.065) -- (8.8,9.);
\draw [] (6.4,10.4) -- (7.082,10.002);
\draw [] (7.254,9.902) -- (7.73,9.624);
\draw [] (7.902,9.524) -- (8.162,9.372);
\draw [] (8.334,9.272) -- (8.8,9.);
\draw [thick] (1.7,9.4) -- (2.9,10.4);
\draw [thick] (5.3,9.) -- (4.1,8.);
\draw [thick] (5.2,9.4) -- (6.4,10.4);
\draw [thick] (8.8,9.) -- (7.6,8.);
\draw [thick] (8.8,10.4) -- (7.6,9.4);
\draw [thick] (1.7,8.) -- (2.607,8.756);
\draw [thick] (2.761,8.884) -- (2.9,9.);
\draw [thin] (1.7,8.) -- (4.1,8.);
\draw [thin] (1.7,9.4) -- (1.7,8.);
\draw [thin] (1.7,9.4) -- (4.1,8.);
\draw [thin] (2.9,9.) -- (5.3,9.);
\draw [thin] (5.2,9.4) -- (7.6,8.);
\draw [thin] (5.2,9.4) -- (8.8,10.4);
\draw [thin] (6.4,10.4) -- (8.8,10.4);
\draw [thin] (7.6,9.4) -- (7.6,8.);
\draw [thin] (8.8,10.4) -- (8.8,9.);
\draw [thin] (2.9,10.4) -- (2.9,9.394);
\draw [thin] (2.9,9.194) -- (2.9,9.);
\draw [thin] (1.7,8.) -- (3.224,8.423);
\draw [thin] (3.416,8.477) -- (5.3,9.);
\node [] at (1.6,9.5) {\scriptsize $1$};
\node [] at (1.6,8.) {\scriptsize $1'$};
\node [] at (2.7,10.4) {\scriptsize $2$};
\node [] at (2.8,9.1) {\scriptsize $2'$};
\node [] at (4.3,8.) {\scriptsize $3'$};
\node [] at (5.3,8.8) {\scriptsize $4'$};
\node [] at (2.1,10.) {\scriptsize $k_1$};
\node [] at (2.3,8.7) {\scriptsize $k_1$};
\node [] at (4.9,8.4) {\scriptsize $k_2$};
\node [] at (1.6,8.8) {\scriptsize $x$};
\node [] at (3.,9.7) {\scriptsize $x$};
\node [] at (2.9,7.8) {\scriptsize $xg\bar y$};
\node [] at (3.2,8.7) {\scriptsize $g\bar y$};
\node [] at (5.1,9.5) {\scriptsize $1$};
\node [] at (6.2,10.4) {\scriptsize $2$};
\node [] at (7.5,9.5) {\scriptsize $3$};
\node [] at (7.8,8.) {\scriptsize $3'$};
\node [] at (9.,10.4) {\scriptsize $4$};
\node [] at (9.,9.) {\scriptsize $4'$};
\node [] at (6.7,9.5) {\scriptsize $g$};
\node [] at (7.5,10.3) {\scriptsize $g$};
\node [] at (5.6,10.) {\scriptsize $k_1$};
\node [] at (8.4,8.4) {\scriptsize $k_2$};
\node [] at (8.9,9.7) {\scriptsize $y$};
\node [] at (8.6,9.5) {\scriptsize $yk_2$};
\node [] at (6.7,8.7) {\scriptsize $g\bar y$};
\node [] at (4.1,9.9) {\scriptsize $g\bar y$};
\node [] at (4.,8.9) {\scriptsize $xg\bar y$};
\node [] at (3.,7.4) {\scriptsize $\frac{\beta(g y \bar g,gk_2 \bar g)\beta_{gk_2\bar g}(gy\bar g,g\bar y)}{\beta(gk_2 \bar g,gy\bar g)}$};
\node [] at (7.,7.4) {\scriptsize $\frac{\beta(k_2,y)}{ \beta_{gk_2\bar g}(g\bar y, y)\beta(y,k_2)}$};
\path [gray, fill, opacity=0.5, opacity=0.3] (5.2,9.4) -- (6.4,10.4) -- (8.8,10.4) -- cycle;
\path [gray, fill, opacity=0.5, opacity=0.3] (5.2,9.4) -- (7.6,9.4) -- (8.8,10.4) -- cycle;
\end{tikzpicture}
}

\begin{document}

\title{Twisted Quantum Double Model of Topological Orders with Boundaries}
\author{Alex Bullivant}
\affiliation{School of Mathematics, University of Leeds, Leeds, LS2 9JT, United Kingdom}
\author{Yuting Hu}
\email{yuting.phys@gmail.com}
\affiliation{Department of Physics and Center for Field Theory and Particle Physics, Fudan University, Shanghai 200433, China}
\author{Yidun Wan}
\email{ydwan@fudan.edu.cn}
\affiliation{Department of Physics and Center for Field Theory and Particle Physics, Fudan University, Shanghai 200433, China}
\affiliation{Collaborative Innovation Center of Advanced Microstructures, Nanjing 210093, China}

\begin{abstract}
We generalize the twisted quantum double model of topological orders in two dimensions to the case with boundaries by systematically constructing the boundary Hamiltonians. Given the bulk Hamiltonian defined by a gauge group $G$ and a three-cocycle in the third cohomology group of $G$ over $U(1)$, a boundary Hamiltonian can be defined by a subgroup $K$ of $G$ and a two-cochain in the second cochain group of $K$ over $U(1)$. The consistency between the bulk and boundary Hamiltonians is dictated by what we call the Frobenius condition that constrains the two-cochain given the three-cocyle. We offer a closed-form formula computing the ground state degeneracy of the model on a cylinder in terms of the input data only, which can be naturally generalized to surfaces with more boundaries. We also explicitly write down the ground-state wavefunction of the model on a disk also in terms of the input data only. 

\end{abstract}
\maketitle

\section{Introduction}\label{sec:intro}
Two-dimensional phases of matter with intrinsic topological orders\cite{Wen1989,Wen1989a,Wen1990a,Wen1990c,Kitaev2003a,Levin2004,Kitaev2006,Chen2012a,Levin2012,Hung2012,Hu2012,Hu2012a,Mesaros2011,Lin2014,Kong2014} have received significant and fast growing attention because of their potential applications in superconductivity\cite{Laughlin1988,Laughlin1988a,Tang2013}, quantum memory\cite{Dennis2002}, and topological quantum computation\cite{Kitaev2003a,Freedman2003,Stern2006,Nayak2008}. Promising candidates of two-dimensional topological orders are such as chiral spin liquids\cite{Kalmeyer1987,Wen1989a}, $\Z_2$ spin liquids\cite{Read1991,Wen1991,Moessner2001}, Abelian quantum Hall states\cite{Klitzing1980,Tsui1982,Laughlin1983}, and non-Abelian fractional quantum Hall states.\cite{TaoWu1984,Moore1991,Wen1991b,Willett1987,Radu2008}

Guided by symmetry considerations, a large class of two-dimensional topological orders can be described and classified by the twisted quantum double model (TQD)\cite{Propitius1995,Hu2012a,Mesaros2011}, which are a Hamiltonian extension of the three-dimensional Dijkgraf-Witten topological gauge theory\cite{Dijkgraaf1990} with finite gauge groups $G$ and $3$-cocyles $\alpha$ in the cohomology group $H^3[G,U(1)]$. In a topological order described by the TQD model on a closed surface with a finite gauge group $G$, anyon excitations carry representations of an emergent, generalized hidden symmetry specified by a quantum group, namely the TQD $D[G]$. The simplest example is the Kitaev model\cite{Kitaev2006}.
Later, the TQD model has also been generalized to three dimensions\cite{Wan2014}. 

Realistic materials that may realize topological orders mostly however have boundaries and thus urge the study of the TQD model on open surfaces. The untwisted version, namely the Kitaev model with boundaries has been studied by in Ref.\cite{Beigi2011} and in Refs.\cite{Cong2016a,Cong2017a}. Two of us also systematically constructed the boundary Hamiltonians\cite{Hu2017} of the Levin-Wen model, which is dual to the construction in this paper in the case with finite groups without cocycle twists. Kitaev and Kong also has a formulation of the gapped boundaries of topological orders in the language of categories\cite{Kitaev2012}, whose relation to the construction in Ref.\cite{Hu2017} is discussed in a parallel paper\cite{Hu2017a} also by two of us. The full TQD model has been studied mostly on closed manifolds, e.g., a torus. Very recently during the preparation of this manuscript, Wang, Wen, and Witten studied the gapped interfaces of symmetric topological orders based on the TQD model\cite{Wang2017}.  

Focus on two dimensions, when there are boundaries, the Hamiltonian of the model would have to contain boundary terms as well. Boundary terms turn affects the spectrum of the model in two aspects. First, a key feature of any topological order---its topologically protected ground state degeneracy (GSD)---may be modified due to its boundary conditions. Second, different boundary conditions correspond to different sets of anyons condensing at the boundaries. These two aspects had been an open problem of topological orders in two dimensions for about two decades until only recently when they were solved for Abelian topological orders\cite{Wang2012} and for general, non-Abelian topological orders\cite{HungWan2014,Lan2014}. Nevertheless, although these solutions do offer a computational method or a counting of the possible boundary conditions and GSD for any given boundary condition, they do not offer a closed-form formula of the GSD counting using solely the topological data of a generic non-Abelian topological order. Moreover, these solutions are abstract rather than being based on certain Hamiltonian model.  

In this work, we generalize the two-dimensional TQD model to the case with boundaries. It is worth of note that when there are no $3$-cocyle twists, the TQD model reduces to the usual Kitaev quantum double (KQD) model, whose boundary terms have been studied by Shor et al\cite{Beigi2011}. In Ref.\cite{Beigi2011}  the boundary conditions are classified by the subgroups of the gauge group $G$ that defines the KQD model. Each subgroup $K\subseteq G$ specifies a boundary anyon condensation. The subgroup $K=\{1_G\}$ with $1_G$ $G$'s identity specifies charge condensation, also known as the rough boundary condition; $K=G$ specifies the flux condensation, also known as the smooth boundary condition; and a $K\subset G$ specifies certain dyon condensation. In the TQD model however, the defining data consists both a gauge group $G$ and a $3$-cocycle $\alpha\in H^3[G,U(1)]$, such that the model describes more topological orders and more exotic anyon spectra than the KQD model does; hence, specifying a subgroup $K\subseteq G$ would not be sufficient for fully characterizing the possibly boundary conditions of a topological order described by the model. It is natural and reasonable to speculate that we need also to specify a $2$-cochain $\beta\in C^2[K,U(1)]$ along with a choice of $K$ because the boundaries are one-dimensional. 

Our strategy is as follows. First, we restrict the boundary degrees of freedom in the TQD model with boundaries to take values in $K\subseteq G$. Second, we add to the original TQD Hamiltonian certain boundary terms depending on $K$ and a $2$-cochain $\beta\in C^2[K,U(1)]$, such that the boundary terms do not affect the exact-solvability of the model. Third, we then study the properties of the model on open surfaces. Our main results are as follows.   

We extend a TQD bulk Hamiltonian by a local boundary Hamiltonian, where the boundary degrees of freedom are in a subgroup (not necessarily a proper one) of the gauge group in the bulk, and the local operators in boundary Hamiltonian are constructed in terms of $2$-cochains of the boundary subgroup. 
The boundary local Hamiltonian needs to be compatible with the bulk Hamiltonian, such that the ground states are invariant under topology-preserving mutation of triangulation both in the bulk and on the boundary. We find that the compatibility condition forms a Frobenius algebra structure on the input $2$-cochain. This agrees with the result in Ref\cite{Hu2017} for the Levin-Wen model with boundaries, which constructs the boundary Hamiltonian in terms of Frobenius algebra from a unitary fusion category.

Base on our boundary Hamiltonian, we write down a formula of the ground-state wavefunction of our model on a disk in terms of the input $2$-cochain only. We also derive a closed-form formula for the GSD on a cylinder in terms of the input data only. We show a couple of examples.
\section{Brief Review of the TQD model}\label{sec:modelRev}
In this section, we briefly review the TQD model on closed surfaces. The TQD model is defined by a infrared fixed point Hamiltonian $H_{G,\alpha}$, with $G$ a finite group and $\alpha\in H^3[G,U(1)]$, on a lattice $\Gamma$ that is a triangulation of a closed $2$D Riemannian surface (Fig. \ref{fig:GraphConfiguration}). Each edge $ab$ between two vertices $a$ and $b$ in $\Gamma$ is graced with a group element $[ab]\in G$, such that the Hilbert space of the model consists of all possible configurations of the group elements on the edges of $\Gamma$. Namely,
\be
\Hil_{\Gamma,G}=\{\{[i j]\in G|i,j\in V(\Gamma)\}\},
\ee
where $V(\Gamma)$ is the set of vertices of $\Gamma$. 
\begin{figure}[!h]
\centering
  \includegraphics[scale=0.6]{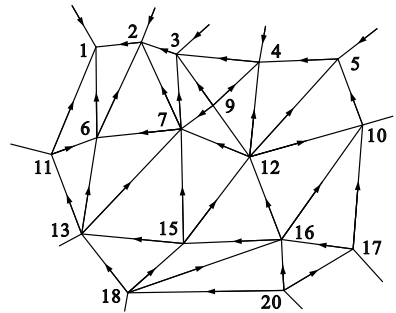}
  \caption{A portion of a graph that represent the basis vectors in the Hilbert space. Each edge carries an arrow and is assigned a group element denoted by $[ab]$ with $a<b$.}
  \label{fig:GraphConfiguration}
\end{figure}
The states are orthogonal in an obvious way. The group elements on the edges can be considered as the discretized gauge field of the underlying Dijkgraaf-Witten topological gauge theory. The graph is oriented with an arbitrary choice of the order of the vertices\footnote{The corresponding triangulation is said to have a branching structure.}, such that each edge is arrowed from its larger vertex to the smaller and  that $[ab]=[ba]^{-1}$, where the exponent $-1$ denote inverse of a group element. Such an ordering of the vertices is called an enumeration\cite{Hu2012}, which does not affect the physics as long as the relative order of the vertices remain unchanged when the graph mutates, i.e., expands or shrinks. The graph $\Gamma$ mutates via the Pachner moves\cite{Pachner1978,Pachner1987} of 2D triangulations, seen in Eq. \eqref{eq:Pachner}.  
\be\label{eq:Pachner}
\begin{aligned}
& f_1:\; \bmm\twoTriangles{}{}{}{}{0}\emm\mapsto\bmm\twoTriangles{}{}{}{}{1}\emm\;\\
& f_2:\; \bmm\oneTriangle{}{}{}\emm\mapsto\bmm\threeTriangles{}{}{}{}{1}\emm\\
& f_3:\; \bmm\threeTriangles{}{}{}{}{1}\emm\mapsto\bmm\oneTriangle{}{}{}\emm.
\end{aligned}
\ee

Certainly the mutation of $\Gamma$ turns $\Gamma$ into a different graph $\Gamma'$ and hence alters the total Hilbert space of the model. But it is shown in Ref.\cite{Hu2012} that the topological properties of the topological order described by the model $H_{G,\alpha}$ remains unchanged because a mutation cannot change the topology of the surface.

For simplicity, we neglect drawing the group elements on the edges but keep only the vertex labels. We may also often refer to $ab$ as an edge or the group element on that edge to avoid clutter. On any part of $\Gamma$ that resembles Fig. \ref{fig:3cocycleA}, one can define a  normalized $3$-cocycle $\alpha(v_1v_2, v_2v_3, v_3v_4)\in H^3[G,U(1)]$. The three variables in the $\alpha$ from left to right are the three group elements, $v_1v_2$, $v_2v_3$ and $v_3v_4$, which are along the path from the least vertex $v_1$ to the greatest vertex $v_4$ passing $v_2$ and $v_3$ in order. Necessary rudiments of such mathematical objects are reviewed in Appendix \ref{app:HnGU1}. Here one should keep in mind that an $\alpha$ is an equivalence class $[\alpha]$ of $U(1)$-valued functions on $G^3=G\x G\x G$. A normalized $\alpha$ is a particular representative of $[\alpha]$ that satisfies the normalization condition 
\begin{align}
  \label{NormalizationCondition}
  \alpha(1,g,h)=\alpha(g,1,h)=\alpha(g,h,1)=1
\end{align}
and the $3$-cocyle condition
\begin{align}
  \label{3CocycleCondition}
  &\alpha(g_1,g_2,g_3)\alpha(g_0\cdot g_1,g_2,g_3)^{-1}\times\\
  &\quad\alpha(g_0,g_1\cdot g_2,g_3)\alpha(g_0,g_1,g_2\cdot g_3)^{-1} \alpha(g_0,g_1,g_2)=1\nonumber
\end{align}
for all $g_i \in G$.

It is shown in Ref.\cite{Hu2012} that each $[\alpha]$ defines a topological order and the choice of the normalized $\alpha$ as the representative is merely a convenience that does not affect the physics. A graph like Fig. \ref{fig:3cocycleA} has a natural signature and hence the associated $3$-cocyle has a chirality determined as follows. One first reads off a list of the three vertices counter-clockwise from any of the three triangles of the defining graph of the $3$--cocycle, e.g., $(v_2,v_3,v_4)$ from Fig. \ref{fig:3cocycleA} and $(v_3,v_2,v_4)$ from Fig. \ref{fig:3cocycleB}. One then append the remaining vertex to the beginning of the list,  e.g., $(v_1,v_2,v_3,v_4)$ from Fig. \ref{fig:3cocycleA} and $(v_1,v_3,v_2,v_4)$ from Fig. \ref{fig:3cocycleB}. If the list can be turned into ascending order by even permutations, such as $(v_1,v_2,v_3,v_4)$ from Fig. \ref{fig:3cocycleA}, one has an $\alpha$ but an $\alpha^{-1}$ otherwise, as by $(v_1,v_3,v_2,v_4)$ from Fig. \ref{fig:3cocycleB}.
In an alternative point of view, if one lifts the vertex $v_2$ in Fig. \ref{fig:3cocycleA} above the paper plane, the three triangles turns out to be on the surface of a tetrahedron. In this sense, one can think of the $3$--cocycle as associated with a tetrahedron as well, and the signature of the graph is the very orientation of the corresponding tetrahedron. This is a useful picture when we evolve the graph $\Gamma$ by the Hamiltonian.
\begin{figure}[h!]
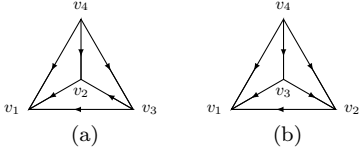

\centering
\subfigure[]{\threeTriangles{v_1}{v_3}{v_2}{v_4}{3} \label{fig:3cocycleA}}
\subfigure[]{\threeTriangles{v_1}{v_2}{v_3}{v_4}{2} \label{fig:3cocycleB}}
\caption{Given the enumeration $v_1<v_2<v_3<v_4$, (a) is the defining graph of the $3$--cocycle $\alpha([v_1v_2], [v_2v_3],[v_3v_4])$, and (b) for $\alpha([v_1v_2], [v_2v_3],[v_3v_4])^{-1}$.}
\label{fig:3cocycle}
\end{figure}
 
The $\alpha$ that defines the model $H_{G,\alpha}$ comprises the matrix elements of the Hamiltonian that reads 
\be\label{eq:Hamiltonian}
  H_{G,\alpha}=-\sum_v A_v-\sum_f B_f,
\ee
where $B_f$ is the face operator defined at each triangular face $f$, and $A_v$ is the vertex operator defined on each vertex $v$. The operator $B_f$ acts on a basis state vector as
\begin{align}
  \label{eq:actionOfBf}
  B_f\BLvert \oneTriangle{v_1}{v_2}{v_3} \Brangle
  =\delta_{[v_1v_2]\cdot[v_2v_3]\cdot[v_3v_1]}
  \BLvert \oneTriangle{v_1}{v_2}{v_3}\Brangle .
\end{align}
The discrete
delta function $\delta_{[v_1v_2]\cdot[v_2v_3]\cdot[v_3v_1]}$ is unity if ${[v_1v_2]\cdot[v_2v_3]\cdot[v_3v_1]=1 }$, where $1$ is the identity element in $G$, and 0 otherwise. Note again that here, the ordering of $v_1,v_2$, and $v_3$ does not matter because of the identities
$\delta_{[v_1v_2]\cdot[v_2v_3]\cdot[v_3v_1]}
=\delta_{[v_3v_1]\cdot[v_1v_2]\cdot[v_2v_3]}$ and
$\delta_{[v_1v_2]\cdot[v_2v_3]\cdot[v_3v_1]}
=\delta_{\{[v_1v_2]\cdot[v_2v_3]\cdot[v_3v_1]\}^{-1}}
=\delta_{[v_3v_1]^{-1}\cdot[v_2v_3]^{-1}\cdot[v_1v_2]^{-1}}
=\delta_{[v_1v_3]\cdot[v_3v_2]\cdot[v_2v_1]}$.
In other words,
in any state on which $B_f=1$ on a triangular face $f$, the three group degrees of freedom
around $v$ is related by a \textit{chain
rule}:
\begin{equation}
\label{eq:chainRule}
[v_1v_3]=[v_1v_2]\cdot[v_2v_3]
\end{equation}
for any enumeration $v_1,v_2,v_3$ of the three vertices of the face $f$. The chain rule \eqref{eq:chainRule} is physically known as the flatness condition in the sense that the gauge connection along the edges of a triangular face is flat. The operator $A_v$ acting on a vertex $v$ is an average
\begin{equation}
  \label{eq:Av}
  A_v=\frac{1}{|G|}\sum_{[vv']=g\in G}A_v^g,
\end{equation}
over the operators $A_v^g$ specified by a group element $g\in G$ acting on the same vertex. The action of $A^g_v$ replaces $v$ by a new enumeration $v'$ that is less than $v$ but greater than all the vertices that are less than $v$ in the original set of enumerations before the action, such that $v'v=g$. In a dynamical language, $v'$ is understood as on the next \textquotedblleft time" slice, and there is an edge $[v'v]\in G$ in the $(2+1)$ dimensional \textquotedblleft spacetime" picture. We illustrate such an action in the example below.
\begin{align}
  \label{eq:Avg}
  &A_{v_3}^g\BLvert \threeTriangles{v_1}{v_2}{v_3}{v_4}{2} \Brangle
  \nonumber\\
  =&\delta_{v'_3v_3,g}
  \alpha\left(v_1v_2,v_2v'_3,v'_3v_3\right)
  \alpha\left(v_2v'_3,v'_3v_3,v_3v_4\right)
  \nonumber\\
  &\times
  \alpha\left(v_1v'_3,v'_3v_3,v_3v_4\right)^{-1}
  \BLvert \threeTriangles{v_1}{v_2}{v'_3}{v_4}{2} \Brangle,
\end{align}
where on the RHS, the new enumerations are in the order $v_1<v_2<v'_3< v_3<v_4$, together with the following flatness conditions.
\be
\begin{aligned}
\label{eq:ChainRuleInBph}
&[v_1{v'_3}]=[v_1v_3]\cdot[v_3{v'_3}],\\
&[v_2{v'_3}]=[v_2v_3]\cdot[v_3{v'_3}],\\
&[{v'_3}v_4]=[{v'_3}v_3]\cdot[v_3v_4].
\end{aligned}
\ee

The basis vector on the LHS of \eqref{eq:Avg} is specified by six group elements,
$[v_1v_3]$, $[v_2v_3]$, $[v_3v_4]$, $[v_1v_4]$, $[v_2v_1]$, and $[v_2v_4]$.The phase factor consisting of three $3$--cocycles on the RHS of Eq. (\ref{eq:Avg}) encodes the non--vanishing matrix elements of $B^{v'_3}_{v_3}$, namely
\be\label{eq:AvMatrix}
\begin{aligned}
&\left(A_{v_3}^g\right)^{[v_1v_3][v_2v_3][v_3v_4]}_{[v_1v'_3] [v_2v'_3][v'_3v_4]}(v_1v_2,v_2v_3,v_1v_3)\\
=&  \alpha\left(v_1v_2,v_2v'_3,v'_3v_3\right)
  \alpha\left(v_2v'_3,v'_3v_3,v_3v_4\right)\\
  &\times
  \alpha\left(v_1v'_3,v'_3v_3,v_3v_4\right)^{-1}.
\end{aligned}
\ee
The $3$-cocycles appearing on the RHS of Eq. (\ref{eq:Avg}) can be easily understood from Fig. \ref{fig:BvTrivalent}.  This figure illustrates the time evolution of the graph before being acted on by $A^g_{v_3}$ to that after the action. We leave more details of this picture to Appendix \ref{app:TQD}.
\begin{figure}[h!]
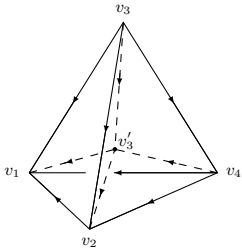

\centering
\BvTri{v_1}{v_2}{v_3}{v_4}{3}
\caption{The topology of the action of $A_{v_3}^g$.}
\label{fig:BvTrivalent}
\end{figure}

The vertex operator in Eq. \eqref{eq:Avg} can naturally extends its definition from a trivalent vertex to a vertex of any valence. The number of $3$--cocyles in the phase factor brought by the action of $A_v^g$ on a vertex is equal to the valence of the vertex. The chirality of each $3$-cocycle in the phase factor follows the criteria described in Appendix \ref{app:TQD}. It is clear that $A^{g=1}_v\equiv \mathbb{I}$ by the normalization of $\alpha$. It is shown that all $B_f$ and $A_v$ are projection operators and commute with each other (see Appendix \ref{app:algAvBf}), which renders the Hamiltonian \eqref{eq:Hamiltonian} exactly solvable. The ground states and all elementary excitations are thus common eigenvectors of these projectors; they carry representations of the TQD $D^\alpha[G]$. On a torus, there is a one-to-one correspondence between the ground state basis states and the types of anyon excitations. More precisely, on a torus, a ground state basis state or its corresponding anyon excitation can be labeled by $\ket{A,\mu}$, where $A$ is a conjugacy class of $G$ and $\mu$ an irreducible representation of the centralizer of $A$ in $G$. This representation $\mu$ is of a special type, called $\beta_{g^A}$-regular, which is explained in Appendix \ref{app:TQD}. This $\beta_{g^A}$ is a twisted $2$-cocycle derived from $\alpha$ via the slant product \eqref{eq:slantProd} introduced in Appendix \ref{app:HnGU1}. Interestingly, the topological orders described by the TQD model is not classified by the $3$-cocycles $\alpha\in H^3[G,U(1)]$ given $G$ but instead classified by the twisted $2$-cocycle $\beta_{g^A}$ derived from $\alpha$.\cite{Hu2012} On a torus, the GSD of the model $H_{G,\alpha}$ is
\be\label{eq:GSDrepresentations}
  \text{GSD}
  =\sum_{A}\#(\beta_{g^A}\text{--representations of }Z^A),
\ee     
where the sum runs over all conjugacy classes of $G$ and $Z^A\subset G$ is the centralizer of the conjugacy class $A$.

It is clear that the TQD model $H_{G,\alpha=1}$ reduces to the usual KQD model, where the action of the vertex operators $A_v^g$ implements gauge transformations on the group elements on the edges incident at the vertex $v$.
\section{TQD model with boundaries}
We now extend the TQD model reviewed in the previous section to one that works on open surfaces. To this end, we need to add boundary terms to the Hamiltonian \ref{eq:Hamiltonian}, preserving the exact solvability of the model. In Ref.\cite{Beigi2011}, for the KQD model on an open square lattice, the boundary operators descend directly from the bulk operators with, however, restricting the boundary gauge fields to take value in a subgroup $K\subseteq G$; different subgroups $K$ characterize different boundary conditions, or equivalently speaking, different boundary anyon condensation. Inspired by this construction, for the TQD model, we can construct the boundary operators likewise. Moreover, the compatibility between the bulk and boundary whose degrees of freedom are restricted to a subgroup $K\subseteq G$ of the TQD model still leaves room for another tweak. Namely, we can associate a $2$-cochain $\beta\in C^2[K, U(1)]$ to the action of a boundary vertex operator. Later, we will show that given a $K\subseteq G$, all possible boundary conditions, each specified by a $\beta$, are in one-to-one correspondence with the $2$-cocycles in $H^2[K,U(1)]$, which generalizes the consideration of a boundary $2$-cocycle in $H^2[K,U(1)]$ in Ref.\cite{Beigi2011} for the KQD model. 
\begin{figure}[h!]
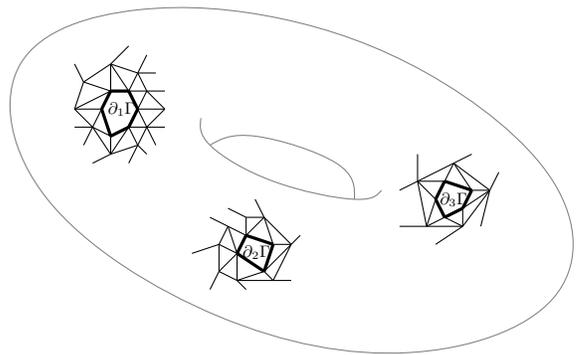

\label{fig:boundaries}
\torusHole
\caption{A torus with multiple holes. Only the lattice near the boundaries are shown explicitly.}
\end{figure}

Let us first write down the general Hamiltonian of the TQD model with multiple disjoint boundaries, followed by explanation.
\be\label{eq:HamWithBoundaries}       
H^{K,\beta}_{G,\alpha}=H_{G,\alpha}-\sum_{i=1}^M\left(\sum_{v\in\partial_i\Gamma}
A_v^{K_i}-\sum_{f\in\partial_i\Gamma}B_f^{K_i}\right),
\ee
where $H_{G,\alpha}$ is the bulk Hamiltonian \eqref{eq:Hamiltonian}, and the rest are the boundary terms. In this general form, we assume the lattice system $\Gamma$ has $M$ boundaries, $\partial_1\Gamma,\partial_2\Gamma,\dots,\partial_M\Gamma$, as sketched in Fig. \ref{fig:boundaries}. Each boundary certainly not necessarily bounds a hole but can be infinitely long, such as a side of an infinite strip. On the $i$-th boundary, the degrees of freedom are restricted to the subgroup $K_i\subseteq G$. A boundary vertex $v$ sits right on the boundary, whereas a boundary triangular face $f$ contains one and only one edge on the boundary and two virtual edges, as in Fig. \ref{fig:bPlaquette}. We now explain the boundary operators individually. Boundary plaquette operators simply project the boundary degrees of freedom to a subgroup $K\subseteq G$:
\be
\begin{aligned}
&B^K_f =\sum_{k\in K} B_f^k,\\
&B^k_f\BLvert\bPlaquette{1}{2}{}{}{}{}\Brangle\\
&=\delta_{[12],k}\BLvert\bPlaquette{1}{2}{}{}{}{}\Brangle.
\end{aligned}
\ee

\begin{figure}[h!]
\label{fig:bPlaquette}
\bPlaquette{1}{2}{}{bulk}{boundary}{f}
\caption{A boundary face $f$ is made of a boundary edge, say, $[12]$ and two virtual edges, the two dotted lines below the boundary. Only a segment of the boundary is shown.}
\end{figure}

The $B^K_p$ operator defined above thus clearly satisfies $B^K_pB^K_p=B^K_p$ and is a projector. The commutativity $[B^K_p,B^K_{p'}]=0$ is also obvious. Hereafter, we shall not draw any virtual boundary face. A boundary segment is always placed horizontally unless stated otherwise, such that the bulk is above the boundary.

The boundary vertex operators acts on the vertices right on a boundary, defined in the example below without loss of generality.
\be\label{eq:edgeAv}
\begin{aligned}
&A^K_v=\frac{1}{|H|}\sum_{k\in K}A^k_v,\\
&A^k_1\BLvert\bPlaquetteTwo{0}{2}{1}{3} \Brangle=\mathcal{A}(01,12,13,k) \BLvert\bPlaquetteTwo{0}{2}{1'}{3} \Brangle,
\end{aligned}
\ee
where in the second line, the vertex $1$ is chosen for illustration of the action, and $1'1:=k$ by definition. Here we only depict two bulk plaquettes because the rest plaquettes are irrelevant to the action of $A^k_1$. The action of $A^k_1$ replaces the boundary vertex $1$ by a new boundary vertex $1'$ with $1'<1$ (This notation is explained in Section \ref{sec:modelRev}.) together with an amplitude $\mathcal{A}(01,12,13,k)$, a $U(1)$ function of the group elements $01\in G$, and $12,13,k=1'1\in K$. 
\be\label{eq:edgeAvAmplitude}
\begin{aligned}
\mathcal{A}(01,12,13,k)&=\bmm\AvAmplitude{2}{1}{3}{0}{1'}\emm\\
&=\frac{\alpha(01',1'1,12)\beta(1'1,12)}{\alpha(01',1'1,13)\beta(1'1,13)}.
\end{aligned}
\ee

Similar to the action of a bulk vertex operator described in the previous section, the action of an edge vertex operator, such as the $A^k_1$ in Eq. \eqref{eq:edgeAv}, evolves the original spatial lattice to a new spatial lattice. Such an evolution creates a spacetime $3$-complex, e.g., the one in Eq. \eqref{eq:edgeAvAmplitude}, which encodes the amplitude of the action. Let us explain the amplitude \eqref{eq:edgeAvAmplitude}. The two $3$-cocycles $\alpha(01',1'1,12)$ and $\alpha(01',1'1,13)^{-1}$ are respectively associated with the tetrahedra $01'12$ and $01'13$ in the $3$-complex in Eq. \eqref{eq:edgeAvAmplitude}. As in the case with bulk vertex operators, the newly created three triangles along the time direction due to the action of $A^k_1$, namely $01'1$, $1'12$, and $1'13$, shaded in Eq. \eqref{eq:edgeAvAmplitude}, must be flat as well, leading to the following chain rules of group elements:
\be
01'\cdot 1'1=01,\quad 1'1\cdot 12=1'2,\quad 1'1\cdot 13=1'3.
\ee 
This is why the amplitude \eqref{eq:edgeAvAmplitude} $\mathcal{A}(01,12,13,k)$ depends only on the original group elements $01,12,13$, and the group element $k=1'1$. A boundary vertex operator differs from a bulk vertex operator by the boundary $2$-cochains in its amplitudes, which we now elaborate on. 

Staring at the figure in the amplitude \eqref{eq:edgeAvAmplitude}, one sees two boundary triangles, $1'12$ and $1'13$, extending along the time direction due to the action of $A^k_1$. This enables the freedom of associating a $U(1)$ factor with each of the two boundary temporal triangles that contributes to the amplitude of $A^k_1$. Such a $U(1)$ factor depends only on the group elements of $K\subseteq G$ on the sides of the corresponding temporal triangle. Since any boundary temporal triangle must satisfy the flatness condition, as it is created by a boundary vertex operator, it inhabits only two independent group elements of $K$. Without any further constraints, hence, such a $U(1)$ factor is a $2$-cochain $\beta\in C^2[K,U(1)]$. Such a $2$-cochain also depends on the orientation of the boundary temporal triangle. The canonical orientation of a triangle on the boundary of a tetrahedron is defined in this way: One grabs using one hand such a triangle along the ascending direction of the vertex labels of the triangle, while keeping the thumb pointing outside of the tetrahedron; if this can only be achieved by one's right (left) hand, the triangle has a positive (negative) orientation. For example, the boundary temporal triangle $1'12$ has positive orientation, whereas $1'13$ has negative orientation; hence, respectively they contribute to the amplitude \eqref{eq:edgeAvAmplitude} $2$-cochains $\beta(1'1,12)$ and $\beta(1'1,13)^{-1}$.

One may wonder why the amplitude of a bulk vertex operator, say Eq. \eqref{eq:Avg} for example, does not contain any $2$-cochains associated with the relevant bulk temporal triangles. The reason is, each bulk temporal triangle belongs to two neighbouring tetrahedra and would thus contribute a $2$-cochain twice to the amplitude but with opposite signs, hence canceling each other. As such, it is only at a boundary the freedom of choosing a $2$-cochain takes effect.

Having introduced the the action of the boundary operators in the new Hamiltonian \eqref{eq:HamWithBoundaries}, we need to check whether these operators are still commuting projectors and their commutativity with the bulk operators. We would leave all such detailed calculations to the appendix but only show in below the commutativity between any two boundary vertex operators because this will lead to the Frobenius condition, which is of paramount importance in this work.     

Consider two boundary vertex operators $A^k_v$ and $A^{k'}_{v'}$. If $v$ and $v'$ are not directly connected by a boundary edge, then obviously $[A^k_v,A^{k'}_{v'}]=0$. Otherwise, let us concretely compute the scenario in Fig. \ref{fig:AvCommute}.
\begin{figure}[h!]
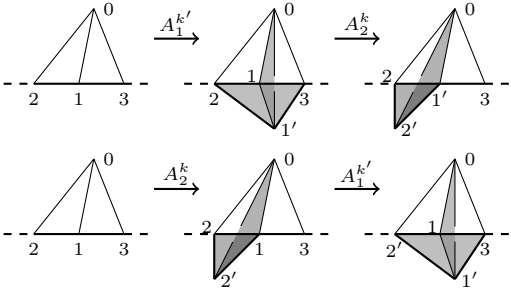

\centering
\AvCommute
\caption{The amplitudes of $A^k_2 A^{k'}_1$ (first row) and $A^{k'}_1 A^k_2$ (second row). Note that $1'1=k'$ and $2'2=k$. Shaded triangles are flat.}
\label{fig:AvCommute}
\end{figure}

We can extract from the spacetime $3$-complexes in Fig. \ref{fig:AvCommute} the following two amplitudes respectively of $A^k_2 A^{k'}_1$ and $A^{k'}_1 A^k_2$.
\begin{align}
A^k_2 A^{k'}_1\BLvert \bPlaquetteTwo{0}{2}{1}{3} \Brangle &=\BLvert \bPlaquetteTwo{0}{2'}{1'}{3} \Brangle \\
&\x \frac{\alpha(01',1'1,12)}{\alpha(01',1'1,13)\alpha(01',1'2',2'2)} \nonumber\\
&\x \frac{\beta(1'1,12)}{\beta(1'1,13)\beta(1'2',2'2)}\nonumber.
\end{align}
\begin{align}
A^{k'}_1 A^k_2 \BLvert \bPlaquetteTwo{0}{2}{1}{3} \Brangle &=\BLvert \bPlaquetteTwo{0}{2'}{1'}{3} \Brangle \\
&\x \frac{\alpha(01',1'1,12')}{\alpha(01,12',2'2)\alpha(01',1'1,13)} \nonumber\\
&\x \frac{\beta(1'1,12')}{\beta(12',2'2)\beta(1'1,13)}\nonumber.
\end{align}
The task now is to demonstrate that the two amplitudes above are equal. It suffices to show that 
\be\label{eq:AvCommute}
 \frac{\alpha(01',1'1,12)\beta(1'1,12)}{\alpha(01',1'2',2'2)\beta(1'2',2'2)}
 = \frac{\alpha(01',1'1,12')\beta(1'1,12')}{\alpha(01,12',2'2)\beta(12',2'2)}.
\ee
Using the $3$-cocycle condition
\be
\frac{\alpha(1'1,12',2'2)\alpha(01',1'2',2'2)\alpha(01',1'1,12')}{\alpha(01,12',2'2) \alpha(01',1'1,12)}=1,
\nonumber
\ee
Eq. \eqref{eq:AvCommute} boils down to the following condition
\be\label{eq:FrobeniusCond}
\alpha(1'1,12',2'2)\frac{\beta(12',2'2)\beta(1'1,12)}{\beta(1'1,12')\beta(1'2',2'2)} =1.
\ee
In other words, if we demand that $[A^k_2,A^{k'}_{1'}]=0$, the above condition must be hold. If not, the Hamiltonian \eqref{eq:HamWithBoundaries} ceases being exactly solvable. Since our purpose is to construct an exactly solvable Hamiltonian with boundaries, we would not consider the possibility of violating the above condition. Condition \eqref{eq:FrobeniusCond} is mathematically known as the Frobenius condition\cite{}, which can also be presented graphically as
\be\label{eq:FrobeniusCondGraph}
\bmm \Frobenius \emm =1.
\ee  
In the equation above, the group elements on the edges all lie in the subgroup $K\subseteq G$, and each triangle is flat. The tetrahedron $1'12'2$ in Eq. \eqref{eq:FrobeniusCondGraph} corresponds to the $3$-cocycle $\alpha(1'1,12',2'2)$ in Eq. \eqref{eq:FrobeniusCond}. The four flat boundary temporal triangles $1'12'$, $1'12$, $1'2'2$, and $12'2$ corresponds respectively to the four $2$-cochains in Eq. \eqref{eq:FrobeniusCond}.  

Here is the essence of the Frobenius condition. The $3$-cocycle $\alpha\in H^3[G,U(1)]$ that defines the model must become cohomologically trivial, i.e., $\alpha|_K \sim\ 1$, when all its three arguments are restricted to the subgroup $K\subseteq G$ because it is equal to a coboundary made of the $2$-cochains $\beta$ in Eq. \eqref{eq:FrobeniusCond}. This strongly constrains what boundary conditions are feasible in the TQD model with certain gauge group $G$. More precisely speaking, the Frobenius condition restricts what subgroups of $G$ can live on a boundary of the model. 

For better understanding of this point, let us consider the simplest example, $D^\alpha[\Z_2]$, the TQD of $G=\Z_2$. Because $H^3[\Z(2),U(1)]=\Z_2$, there are only two such models. One is the $\Z_2$-toric code defined by $\alpha\equiv 1$. The other is the doubled semion model defined by $\alpha(g_1,g_2,g_3)=-1$ if $g_1,g_2,g_3\neq 1$ otherwise $\alpha=1$. Therefore, for the $\Z_2$-toric code, the only two subgroups of $\Z_2$, namely the trivial group $\{1\}$ or the entire $\Z_2$, can be legal boundary conditions, as $\alpha\equiv 1$. That is, the $\Z_2$-toric code has two possible boundary conditions. Nevertheless, for the doubled semion model, in order to satisfy the Frobenius condition, only $K=\{1\}$ is allowed to exist a boundary. That is, the doubled semion has a unique boundary condition. We would get back to this example again later in the paper.  

Having shown that the boundary vertex operators commute with each other, we also need to show that they are projectors, namely $A^K_v A^K_v=A^K_v$ $\forall v\in\partial\Gamma$. For this to hold, it suffices to show that $A^{k'}_v A^k_v=A^{k'k}_v$. This can be done using the $3$-cocyle condition \eqref{3CocycleCondition} and the Frobenius condition \eqref{eq:FrobeniusCond}, the detail of which is left to Appendix \ref{appd}. 

As such, the Hamiltonian \eqref{eq:HamWithBoundaries} is again exactly solvable and composed of projectors. We can then place the model on surfaces with boundaries to study its physical properties, as what we are going to do shortly. Before that, let us prove the following theorem, as promised earlier.
\begin{theorem}
Given a $K\subseteq G$, the $2$-cochain solutions $\beta\in C^2[K,U(1)]$ to the Frobenius condition \eqref{eq:FrobeniusCond} are in one-ton-one correspondence with the $2$-cocycle in $H^2[K,U(1)]$.
\end{theorem} 
\begin{proof}
Among all possible solutions $\beta$ to Eq. \eqref{eq:FrobeniusCond}, let us take an arbitrary one and call it $\beta_0$. Let $H^2[K,U(1)]=\{\tilde\beta_1, \tilde\beta_2,\dots,\tilde\beta_n\}$, where $\tilde\beta_i$ and $\tilde\beta_j$ are in equivalent $2$-cocycles. Because $\dd \tilde\beta_i\equiv 1,\ \forall i$ and $\dd(\beta\tilde\beta_i)=\dd \beta\dd \tilde\beta_i$, for any $\beta_0$, the $H^2[K,U(1)]$ yields a set of solutions $\{\beta_0,\beta_0\tilde\beta_1, \beta_0 \tilde\beta_2,\dots,\beta_0\tilde\beta_n\}$ to Eq. \eqref{eq:FrobeniusCond}.

Conversely, consider any other solution $\beta_m$ to Eq. \eqref{eq:FrobeniusCond}, we have $\alpha\dd\beta_m =\alpha\dd\beta_0=1$. Hence, 
\[ 
\dd\beta_m=\dd\beta_0\Rightarrow \dd(\beta_m\beta_0^{-1})=1\Rightarrow \beta_m =\beta_0\tilde\beta_m, 
\]
where $\tilde\beta_m\in H^2[K,U]$ is a $2$-cocycle. The one-to-one correspondence is thus established. And it does not matter which solution $\beta_0$ we choose to generate the set of solutions $\{\beta_0,\beta_0\tilde\beta_1, \beta_0 \tilde\beta_2, \dots,\beta_0\tilde\beta_n\}\xeq{\mathrm{def}}\{\beta_i|i=0,\dots, n=|H^2[K,U(1)]|-1\}$. For future convenience, we denote this set of $2$-cochains that specify all possible boundary conditions for a given $K\subseteq G$ by $\Lambda_K$.  
\end{proof}  

Now the question is whether two pairs $(\alpha,\beta)$ and $(\alpha',\beta')$ , where $\alpha,\alpha'\in H^3[G,U(1)]$ and $\beta,\beta'\in \Lambda_K$, define the same TQD model with a boundary. 

\section{GSD on a cylinder}
We first consider the first nontrivial case, namely a sphere with two holes, which is homeomorphic to a cylinder. We shall place our model on the cylinder (Fig. \ref{fig:cylinder}) and assume the two ends of the cylinder may respectively possess boundary conditions specified by subgroups $K_1$ and $K_2$ of the gauge group $G$. The two subgroups $K_1$ and $K_2$ may or may not be the same. 
\begin{figure}[h!]
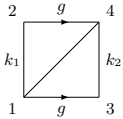

\cylinderGraph{g}{g}{k_1}{k_2}{1}
\caption{Lattice triangulation of a cylinder. The arrows indicate the identified edges. The group elements are $g\in G$, $k_1\in K_1\subseteq G$, and $k_2\in K_2\subseteq G$.}
\label{fig:cylinder}
\end{figure}

Since now we are interested in ground states only, we can assume flatness on the two triangles in Fig. \ref{fig:cylinder}. That is, we are working in the subspace $\Hil^{B_f=1}$ of the total Hilbert space. Hence, the group element on the diagonal edge $14$ in Fig. \ref{fig:cylinder} is determined by the group elements on the horizontal and vertical edges in the figure. Note however that we have both $14=k_1 g$ and $14=g k_2$; hence, $k_1=g k_2 \bar g$, i.e. $k_1$ and $k_2$ despite in possibly different subgroups of $G$ still belong to the same conjugacy class of $G$ in the ground state space.

In the subspace $\Hil^{B_f=1}$, also because the vertices in Fig. \ref{fig:cylinder} are all boundary vertices, the ground state projector reduces to 
\be
\label{eq:cylGSDproj}
P^0_{cyl}=\Pi_{v\in\partial\Gamma}A_v=
\frac{1}{|K_1||K_2|}\sum_{x\in K_1} A^{x}_{1=2} \sum_{y\in K_2} A^{y}_{3=4},
\ee
where in fact vertices $1$ and $2$ are identified, and vertices $3$ and $4$ are identified. Note that when act the above operator on the cylinder, one still needs to act $A^x$ on vertices $1$ and $2$ individually, as if $1$ and $2$ are different vertices; however, the identification of $1$ and $2$ will be automatically accounted for by the periodic boundary condition and that there is only one normalization factor $1/|K_1|$. The same procedure applies to $A^y$ on vertices $3$ and $4$. The GSD of our model on the cylinder in Fig. \ref{fig:cylinder} thus can be obtained by
\be
\begin{aligned}
& GSD_{cyl} \\
=& \mathrm{Tr}_{\Hil^{B_f=1}} P^0_{cyl} \\
=& \sum_{\scriptsize \substack{k_1\in K_1,\\ k_2\in K_2,\\ g\in G}}\delta_{k_1, gk_2\bar g} \Blangle\cylinderGraph{g}{g}{k_1}{k_2}{0.8} \BRvert P^0_{cyl}\BLvert\cylinderGraph{g}{g}{k_1}{k_2}{0.8} \Brangle
\end{aligned}
\ee

To obtain a concrete answer, we first check how the projector $P^0_{cyl}$ acts on the cylinder. The order of acting the vertex operators comprising $P^0_{cyl}$ is irrelevant because they commute. To simply the calculation, however, we choose to act on the vertices in descending order. The entire action creates a spacetime $3$-complex shown in Fig. \ref{fig:cylinderGSD}, from which we can extract the amplitude of $P^0_{cyl}$ as follows.
\begin{align}
P^0_{cyl}\ket{12,13,34} &=\frac{1}{|K_1||K_2|}\sum_{x\in K_1}\sum_{y\in K_2} \ket{1'2',1'3',3'4'}\label{eq:cylGSDprojAmp} \\
&\x \frac{\alpha(13,34',4'4)\beta(34',4'4)}{\alpha(12,24',4'4)\alpha(13',3'3,34') \beta(3'3,34')}\nonumber \\
&\x\frac{\alpha(12',2'2,24') \alpha(1'1,13',3'4')\beta(1'1,12')}{\alpha(1'1,12',2'4') \beta(12',2'2)}, \nonumber
\end{align}
where take a simpler notation for the state on the cylinder, namely $\ket{12,13,34}$ and $\ket{1'2',1'3',3'4'}$ for the initial and final states, and $1'1=2'2=x\in K_1$, $3'3=4'4=y\in K_2$.

\begin{figure}[h!]
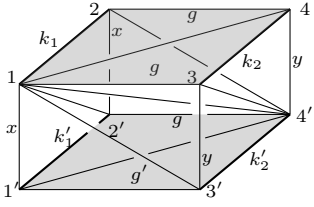

\cylinderGSD
\caption{The spacetime $3$-complex created by acting $P^0_{cyl}$ on the cylinder in Fig. \ref{fig:cylinder}. The top and bottom layers are respectively the original and the new lattices of the cylinder before and after the action. Time direction is downward. All the triangles are flat.}
\label{fig:cylinderGSD}
\end{figure}  

Using the flatness condition on all the triangles in Fig. \ref{fig:cylinderGSD}, in terms of the group elements explicitly, Eq. \eqref{eq:cylGSDprojAmp} becomes
\begin{align}
P^0_{cyl}\ket{k_1,g,k_2} &=\frac{1}{|K_1||K_2|}\sum_{x\in K_1}\sum_{y\in K_2} \ket{xk_1\bar x, xg\bar y, yk_2\bar y}\nonumber \\
&\x \frac{\alpha(g,k_2\bar y,y)\beta(k_2\bar y,y)}{\alpha(k_1,g\bar y,y)\alpha(g\bar y,y,k_2\bar y) \beta(y,k_2\bar y)}\label{eq:cylGSDprojAmpExpl} \\
&\x\frac{\alpha(k_1\bar x,x,g\bar y) \alpha(x,g\bar y,yk_2\bar y)\beta(x,k_1\bar x)}{\alpha(x,k_1\bar x,xg\bar y) \beta(k_1\bar x,x)}. \nonumber
\end{align}

The GSD on the cylinder in can then be obtained as
\begin{align}
GSD_{cyl}&=\sum_{\substack{k_1\in K_1 \\k_2\in K_2}}\sum_{g\in G}\delta_{k_1,gk_2\bar g} \bra{k_1,g,k_2} P^0_{cyl}\ket{k_1,g,k_2} \nonumber \\
&= \sum_{\substack{x,k_1\in K_1 \\y,k_2\in K_2}}\sum_{g\in G} \frac{\delta_{k_1,gk_2\bar g}\delta_{k_1x, xk_1}\delta_{k_2y,yk_2}\delta_{xg,gy}}{|K_1||K_2|}  \nonumber \\
&\x \frac{\alpha(g,k_2\bar y,y)\beta(k_2\bar y,y)}{\alpha(k_1,g\bar y,y)\alpha(g\bar y,y,k_2\bar y) \beta(y,k_2\bar y)}\nonumber \\
&\x\frac{\alpha(k_1\bar x,x,g\bar y) \alpha(x,g\bar y,yk_2\bar y)\beta(x,k_1\bar x)}{\alpha(x,k_1\bar x,xg\bar y) \beta(k_1\bar x,x)}.
\end{align}
The above expression can be massaged into a more compact and topologically meaningful form by the following procedure. First, using the $\delta$-functions above and the relations $x=gy\bar g$, $xg\bar y= g$, $h_1=gk_2 g^{-1}$, and $x=g yg^{-1}$ implied by the flatness conditions, $GSD_{cyl}$ becomes
\begin{align*}
GSD_{cyl}
&= \sum_{\substack{x,k_1\in K_1 \\y,k_2\in K_2}}\sum_{g\in G} \frac{\delta_{k_1,gk_2\bar g}\delta_{k_1x, xk_1}\delta_{k_2y,yk_2}\delta_{xg,gy}}{|K_1||K_2|}   \\
&\x \frac{\alpha(g,k_2\bar y,y)\beta(k_2\bar y,y)}{\alpha(k_1,g\bar y,y)\alpha(g\bar y,y,k_2\bar y) \beta(y,k_2\bar y)} \\
&\x\frac{\alpha(gk_2\bar y\bar g,gy\bar g,g\bar y) \alpha(gy\bar g,g\bar y,k_2)\beta(gk_2\bar y\bar g,g y \bar g)}{\alpha(gy\bar g,gk_2\bar y\bar g,g) \beta(gy\bar g,gk_2\bar y\bar g)}.
\end{align*} 
Applying the $3$-cocyle condition
\[
\frac{\alpha(gk_2\bar y\bar g,gy\bar g,g\bar y)\alpha(gy\bar g,gk_2\bar g,g\bar y)\alpha(gy\bar g,gk_2\bar y\bar g,gy\bar g)}{\alpha(gk_2\bar g,gy\bar g,g\bar y)\alpha(gy\bar g,gk_2\bar y\bar g,g)}=1
\]
and the Frobenius condition
\[
\alpha(gy\bar g,gk_2\bar y\bar g,gy\bar g)^{-1} \frac{\beta(gy\bar g,gk_2\bar y\bar g)\beta(gk_2\bar g,gy\bar g)}{\beta(gk_2\bar y\bar g,gy\bar g)\beta (gy\bar g,gk_2\bar g)} =1,
\]
we obtain
\begin{align*}
GSD_{cyl}
&= \sum_{\substack{x,k_1\in K_1 \\y,k_2\in K_2}}\sum_{g\in G} \frac{\delta_{k_1,gk_2\bar g}\delta_{k_1x, xk_1}\delta_{k_2y,yk_2}\delta_{xg,gy}}{|K_1||K_2|}   \\
&\x \frac{\alpha(g,k_2\bar y,y)\beta(k_2\bar y,y)}{\alpha(k_1,g\bar y,y)\alpha(g\bar y,y,k_2\bar y) \beta(y,k_2\bar y)} \\
&\x\frac{\alpha(gk_2\bar g,gy\bar g,g\bar y) \alpha(gy\bar g,g\bar y,k_2)\beta(g y \bar g,gk_2 \bar g)}{\alpha(gy\bar g,gk_2 \bar g,g\bar y) \beta(gk_2 \bar g,gy\bar g)}.
\end{align*} 
Again, by the $3$-cocycle condition
\[
\frac{\alpha(g \bar y,y,k_2)\alpha(g,k_2 \bar y, y)}{\alpha(g \bar y,k_2 y)\alpha(g \bar y, y, k_2 \bar y)\alpha(y,k_2 \bar y,y)}=1
\]
and the Frobenius condition
\[
\alpha(y, k_2 \bar y, y)\frac{\beta(k_2 \bar y, y)\beta(y,k_2)}{\beta(y,k_2 \bar y)\beta(k_2,y)}=1,
\]
we have
\begin{align*}
GSD_{cyl}
&= \sum_{\substack{x,k_1\in K_1 \\y,k_2\in K_2}}\sum_{g\in G} \frac{\delta_{k_1,gk_2\bar g}\delta_{k_1x, xk_1}\delta_{k_2y,yk_2}\delta_{xg,gy}}{|K_1||K_2|}   \\
&\x \frac{\alpha(g\bar y,k_2,y)\beta(k_2,y)}{\alpha(gk_2\bar g,g\bar y,y)\alpha(g\bar y,y,k_2) \beta(y,k_2)} \\
&\x\frac{\alpha(gk_2\bar g,gy\bar g,g\bar y) \alpha(gy\bar g,g\bar y,k_2)\beta(g y \bar g,gk_2 \bar g)}{\alpha(gy\bar g,gk_2 \bar g,g\bar y) \beta(gk_2 \bar g,gy\bar g)}.
\end{align*} 
Now in the expression above, we can apply the definition of twisted $2$-cocycles to the two groups of $3$-cocycles, namely,
\[
\frac{\alpha(g\bar y,k_2,y)}{\alpha(gk_2\bar g,g\bar y,y)\alpha (g\bar y,y,k_2) } =\beta_{gk_2\bar g}(g\bar y, y)^{-1}
\] 
and
\[
\frac{\alpha(gk_2\bar g,gy\bar g,g\bar y) \alpha(gy\bar g,g\bar y,k_2)}{\alpha(gy\bar g,gk_2 \bar g,g\bar y) }=\beta_{gk_2\bar g}(gy\bar g,g\bar y).
\]
Finally, we obtain
\begin{align}
&GSD_{cyl} \nonumber\\
=& \sum_{\substack{x,k_1\in K_1 \\y,k_2\in K_2}}\sum_{g\in G} \frac{\delta_{k_1,gk_2\bar g}\delta_{k_1x, xk_1}\delta_{k_2y,yk_2}\delta_{xg,gy}}{|K_1||K_2|} \label{eq:GSDcylSimple} \\
&\x  \frac{\beta(g y \bar g,gk_2 \bar g)\beta_{gk_2\bar g}(gy\bar g,g\bar y)}{\beta(gk_2 \bar g,gy\bar g)} \frac{\beta(k_2,y)}{ \beta_{gk_2\bar g}(g\bar y, y)\beta(y,k_2)}. \nonumber
\end{align}

The $3$-cocycles and Frobenius conditions applied in the procedure above to obtain Eq. \eqref{eq:GSDcylSimple} are in fact topological moves that turn the triangulation in Fig. \ref{fig:cylinderGSD} into the following triangulation in Fig. \ref{fig:cylinderGSDsimple}.
\begin{figure}[h!]
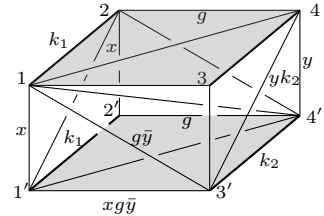

\cylinderGSDsimple
\caption{Triangulation of the cube obtained from that in Fig. \ref{fig:cylinderGSD} by the topological moves associated with the $3$-cocycle and Frobenius conditions used to obtain Eq. \eqref{eq:GSDcylSimple}.}
\label{fig:cylinderGSDsimple}
\end{figure}
 
\begin{figure}[h!]
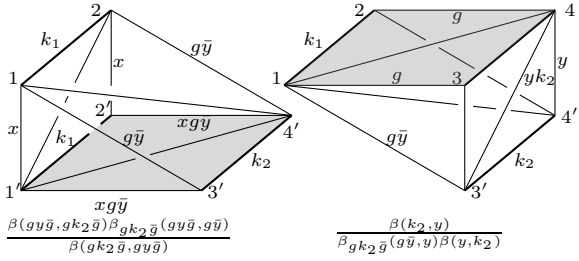

\cylinderGSDsplit
\caption{Splitting the cube in Fig. \ref{fig:cylinderGSDsimple} along the plane $123'4'$ into two halves. Each half is associated with a combination of two $2$-cochains and one twisted $2$-cocycle, presented below the corresponding figure. }
\label{fig:cylinderGSDsplit}
\end{figure}

\section{Ground states on a disk}

Here, we present an explicit formula of the ground-state wavefunction on a disk.

Since we are interested in boundary theories only, we assume no quasiparticles existing in the bulk. Any triangulation of a disk can be reduced to a pie-disk using Pachner moves in the bulk, such that there is only one vertex left in the bulk. See Fig. \ref{fig:Pidisk}.
\begin{figure}[!ht]
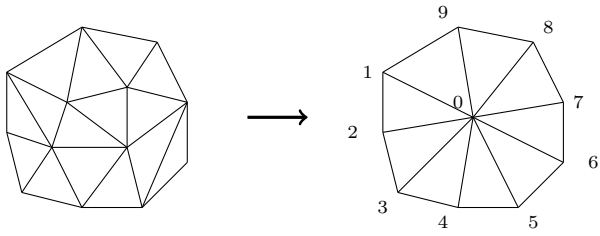

        \centering
        \TriangulationMuttation
        \caption{Any triangulation of a disk (left) can be reduced to a pie-disk using the Pachner moves in the bulk.}
        \label{fig:Pidisk}
\end{figure}

Boundary theories can be studied on the reduced triangulation with $N$ boundary vertices (and $N$ triangles in the bulk). We denote the basis of the Hilbert space by
\begin{equation}\label{eq:Diskbasis}
\{\ket{a_1,a_2,\dots,a_N}\}
\end{equation}
with $a_n=[0n]$ in Fig. \ref{fig:Pidisk}. The boundary edges are determined by these $a_n$.

The ground state is expressed by
\begin{align}\label{eq:GSwavefunction}
\Phi(\{a_n\})=\beta(a_1,a^{-1}_1a_N)\prod_{n=1}^{N-1}\beta(a_n,a^{-1}_na_{n+1})^{-1}.
\end{align}

To check $\ket{\Phi}$ is a ground state, the action of $A_v^k$ in the local basis is
\begin{align}\label{eq:AvkLocalBase}
&A^k_v
\ExpGSbase
\nonumber\\
=
&\frac{\alpha(a_1,a_1^{-1}a_2k^{-1},k)\beta(a_1^{-1}a_2k^{-1})}
{\alpha(a_2k^{-1},k,a_2^{-1}a_3)\beta(k,a_2^{-1}a_3)}
\ExpGSbaseAA
\end{align}
Hence acting $A_v^k$ on $\ket{\Phi}$ in the local basis yeilds
\begin{align}\label{eq:AvkLocalBaseAA}
&A^k_v\left(\beta(a_1,a_1^{-1}a_2)\beta(a_2,a_2^{-1}a_3)\ExpGSbase\right)
\nonumber\\
=
&\beta(a_1,a_1^{-1}ka_2)\beta(ka_2,a_2^{-1}k^{-1}a_3)\ExpGSbaseAA
\end{align}
from which one checks that $\ket{\Phi}$ is a ground state.

The topological feature on a boundary can be described in terms of unitary 1+1D boundary Pachner moves, which expand or shrink the boundary by one boundary edge. See Fig. \ref{fig:BdryPachnerMove}.  The unitary representation of these moves can be constructed in terms of $\beta$, in a way similar to that in Ref\cite{Hu2017}, which defines generic transformations in the language of Frobenius algebra. The ground-state Hilbert space is invariant under these boundary Pachner moves.
\begin{figure}
        \centering
        \BdryPachnerMove
        \caption{Boundary Pachner Moves mutating between two boundary edges and one boundary edges. }
        \label{fig:BdryPachnerMove}
\end{figure}

\section{Examples: $G=D_3$ and $D_4$}

For $D_m$ group, we denote the group elements by $(A,a)$, with $A=0,1$ and $a=0,1,\dots,m-1$. The multiplication is given by
\begin{equation}\label{eq:Dmmultiplication}
(A,a)(B,b)=[A+B]_2[(-1)^{B}a+b]_m
\end{equation}
where $[x]_y:=x\text{ mod }y$.

All 3-cocycles are given by
\begin{widetext}
\begin{align}\label{eq:Dm3cocycles}
\alpha^p\left\{(A,a),(B,b),(C,c)\right\}
=\exp\left\{
\frac{2 i p\pi}{m^2}
\left[
(-1)^{B+C}a 
\biggl(
(-1)^C b + c
-[(-1)^C b+c]_m
\biggr)
        +\frac{m^2}{2} A B C\right]
\right\}
\end{align}
\end{widetext}
where $p=0,1,\dots,2m-1$.

For $D_3$ ($D_4$) group, the number of 2-cochain solutions to the Frobenius condition \eqref{eq:FrobeniusCond} for each 3-cocycle  $\alpha^p$ and each subgroup $K$ is listed in Tab. \ref{tab:D3sols} (Tab. \ref{tab:D4sols}). Since all sub groups of $D_3$ are cyclic, there is at most one solution  for any $\alpha^p$ and $K$. A solution exists if and only if $\alpha^p|_K$ is trivial. And all such solutions are $\beta=1$, using Eq. \eqref{eq:Dm3cocycles}.

\begin{table}
        \caption{Number of 2-cochain solutions for $D_3$}
        \label{tab:D3sols}
        \begin{tabular}{|c|c|c|c|c|c|c|}
                \hline
                $K$ & $\alpha^0$ & $\alpha$ & $\alpha^2$ & $\alpha^3$ & $\alpha^4$ & $\alpha^5$
                \\
                \hline
                $\{00\}$ & 1 & 1 & 1 &  1 & 1 & 1
                \\
                \hline
                $\{00,10\}=Z_2$ & 1 & $\times$ & 1 & $\times$ & 1 & $\times$
                \\
                \hline
                $\{00,11\}=Z_2$ & 1 & $\times$ & 1 & $\times$ & 1 & $\times$
                \\
                \hline
                $\{00,12\}=Z_2$ & 1 & $\times$ & 1 & $\times$ & 1 & $\times$
                \\
                \hline
                $\{00,01,02\}=Z_3$ & 1 & $\times$ &  $\times$ & 1 & $\times$ & $\times$
                \\
                \hline
                $D_3$ & $1$ & $\times$ &  $\times$ & $\times$ & $\times$ & $\times$
                \\
                \hline          
        \end{tabular}
\end{table}

\begin{table}
        \caption{Number of 2-cochain solutions for $D_4$}
        \label{tab:D4sols}
        \begin{tabular}{|c|c|c|c|c|c|c|c|c|}
                \hline
                $K$ & $\alpha^0$ & $\alpha$ & $\alpha^2$ & $\alpha^3$ & $\alpha^4$ & $\alpha^5$ & $\alpha^6$& $\alpha^7$
                \\
                \hline
                $\{00\}$ & 1 & 1 & 1 &  1 & 1 & 1 & 1  & 1
                \\
                \hline
                $\{00,02\}=Z_2$ & 1 & $\times$ & 1 & $\times$ & 1 & $\times$ & 1 & $\times$
                \\
                \hline
                $\{00,10\}=Z_2$ & 1 & $\times$ & 1 & $\times$ & 1 & $\times$ & 1 & $\times$
                \\
                \hline
                $\{00,11\}=Z_2$ & 1 & $\times$ & 1 & $\times$ & 1 & $\times$ & 1 & $\times$
                \\
                \hline
                $\{00,12\}=Z_2$ & 1 & $\times$ & 1 & $\times$ & 1 & $\times$ & 1 & $\times$
                \\
                \hline
                $\{00,13\}=Z_2$ & 1 & $\times$ & 1 & $\times$ & 1 & $\times$ & 1 & $\times$
                \\
                \hline  
                $\{00,01,02,03\}=Z_4$ & 1 & $\times$ & $\times$ & $\times$ & 1 & $\times$ & $\times$ & $\times$
                \\
                \hline
                $\{00,02,10,12\}=Z_2\times Z_2$ & $Z_2$ & $\times$ & $Z_2$ & $\times$ & $Z_2$ & $\times$ & $Z_2$ & $\times$
                \\
                \hline
                $\{00,02,11,13\}=Z_2\times Z_2$ & $Z_2$ & $\times$ & $Z_2$ & $\times$ & $Z_2$ & $\times$ & $Z_2$ & $\times$
                \\
                \hline
                $D_4$ & $Z_2$ & $\times$ & $\times$ & $\times$ & $\times$ & $\times$ & $\times$ & $\times$
                \\
                \hline
        \end{tabular}
\end{table}

\section*{Acknowledgements}
We thank Yong-shi Wu and Ling-Yan Hung for helpful discussions. YDW is also supported by the Shanghai Pujiang program.

\begin{appendix}
\section{Review of $H^n(G,U(1))$}\label{app:HnGU1}
We give a brief review of the cohomology groups $H^n[G,U(1)]$ of finite groups $G$.

We first define the $n$-th \textit{cochain group} $C^n[G,U(1)]$ of $G$, which is an Abelian group of $n$-\textit{cochains}. The group elements of $C^n[G,U(1)]$ are functions $c(g_1,\dots,g_n): G^{\times n}\to U(1)$, where $g_i\in G$. The group multiplication reads $c(g_1,\dots,g_n)c'(g_1,\dots,g_n)=(cc')(g_1,\dots,g_n)$. 
There \textit{is a coboundary operator} $\delta$  that maps $C^n$ to $C^{n+1}$, namely,
\begin{align*}
\delta :C^n\to C^{n+1}: c(g_1,\dots,g_n)\mapsto \delta c(g_0,g_1\dots,g_n),
\end{align*}
where
\begin{align*}
\delta c(g_0,g_1\dots,g_n)= \prod_{i=0}^{n+1}c(\dots,g_{i-2},g_{i-1}g_i,g_{i+1},\dots)^{(-1)^i}.
\end{align*}
At $i=0$, the series of variables starts at $g_0$, and at $i=n+1$, the series of variables ends at $g_{n-1}$. The coboundary operator $\delta$ is nilpotent: $\delta^2c=1$, which results in an exact sequence:
\be\label{eq:exactSeq}
\cdots C^{n-1}\stackrel{\delta}{ \to} C^n\stackrel{\delta}{ \to} C^{n+1}\cdots.
\ee 
The images of the coboundary operator, $\mathrm{im}(\delta:C^{n-1}\to C^n)$, form the $n$-th coboundary group, where the $n$-cochains are called $n$-\textit{coboundaries}. The kernel $\ker(\delta:C^n\to C^{n+1})$ forms the group of $n$-\textit{\textbf{cocycles}}, which are the $n$-cochains meeting the \textit{$n$-cocycle condition} $\delta c=1$. The exact sequence \eqref{eq:exactSeq} leads to the definition of the $n$-th cohomology group:
\[
H^n[G,U(1)]:=\frac{\ker(\delta:C^n\to C^{n+1})}{\mathrm{im}(\delta:C^{n-1}\to C^n)}.
\]
The group $H^n[G,U(1)]$ is clearly Abelian and consists of the equivalence classes of the $n$-cocyles that differ from each other by merely an $n$-coboundary.
A trivial $n$-cocycle is one that can be written as a $n$-coboundary. One can define \emph{a slant product} that maps an $n$-cocycle $c$ to an $(n-1)$-cocycle $c_g$:
\be\label{eq:slantProd}
\begin{aligned}
&c_g(g_1,g_2,\dots,g_{n-1})\\
=&c(g,g_1,g_2,\dots,g_{n-1})^{(-1)^{n-1}}\\
&\times \prod^{n-1}_{j=1} c\left[g_1,\dots,g_j,(g_1\cdots g_j)^{-1}\right. \\
& \left. g(g_1\cdots g_j),\dots,g_{n-1}\right]^{(-1)^{n-1+j}}.
\end{aligned}
\ee

The twisted $2$-cocycles defined above Eq. \eqref{eq:GSDcylSimple} are examples of the slant product above.
\end{appendix}

\bibliographystyle{apsrev}
\bibliography{StringNet}

\end{document}